\documentclass[reprint,amsfonts, amssymb, amsmath, aps, showkeys,prl, superscriptaddress, twocolumn]{revtex4-2}

\usepackage[full]{complexity}
\usepackage{verbatim}

\usepackage{cmd/laro}
\renewcommand{\vec}[1]{\boldsymbol{#1}}
\def\spn{{\rm span}}
\def\sg{\sigma}
\def\va{\alpha}
\def\g{\gamma}
\def\w{\omega}
\def\SU{\textsf{SU}}
\def\GL{\textsf{GL}}

\def\U{\textsf{U}}
\def\triv{{\rm triv}}

\newtheorem{fact}{Fact}
\def\poly{{\rm poly}}

\usepackage{ytableau}
\usepackage{algorithm}
\usepackage{algpseudocode}
\ytableausetup{smalltableaux}
\newcommand\encircle[1]{%
  \tikz[baseline=(X.base)] 
    \node (X) [draw, shape=circle, inner sep=0] {\strut #1};}

\begin{document}

\author{Mart\'{i}n Larocca}
\affiliation{Los Alamos National Laboratory, Los Alamos, New Mexico, USA}
\author{Vojtech Havlicek}
\affiliation{IBM Quantum, IBM T.J. Watson Research Center, New York, USA}

\begin{abstract}
Kostka, Littlewood-Richardson, Plethysm and Kronecker coefficients are the multiplicities of irreducible representations in the decomposition of representations of the symmetric group that play an important role in representation theory, geometric complexity and algebraic combinatorics. We give quantum algorithms for computing these coefficients whenever the ratio of dimensions of the representations is polynomial and study the computational complexity of this problem. We show that there is an efficient classical algorithm for computing the Kostka numbers under this restriction and conjecture the existence of an analogous algorithm for the Littlewood-Richardson coefficients.  We argue why such classical algorithm does not straightforwardly work for the Plethysm and Kronecker coefficients and conjecture that our quantum algorithms lead to superpolynomial speedups for these problems. The  conjecture about Kronecker coefficients was disproved by Greta Panova in [arXiv:2502.20253] with a classical solution which, if optimal, points to a $\OC(n^{4+2k})$ vs $\tilde{\Omega}(n^{4k^2+1})$  polynomial gap in quantum vs classical computational complexity for an integer parameter $k$.
\end{abstract}

\title{Quantum Algorithms for Representation-Theoretic Multiplicities}
\maketitle

\begin{table*}
\centering
\renewcommand{\arraystretch}{1} 
\caption{\textbf{Quantum Algorithms for Branching Coefficients.} Choices of groups $H\subseteq G$, $H$-representation $R$, and an $H$-irrep $r^\va_H$ for which the algorithm in Thm.~\ref{theorem_quantum_algo} computes the studied multiplicity coefficients. Here, $n,a,b,c,d\in \mbb{N}$ where $a+b=n$ (Littlewood-Richardson) and $cd=n$ (Plethysm).  
 The cost $\dim(R)/\dim(r^{\va}_H)$ is the number of samples required to (exactly) compute the coefficient. 
}
\begin{tabular}{|l @{\hspace{1em}} c @{\hspace{1em}} c @{\hspace{1em}} c @{\hspace{1em}} c @{\hspace{1em}} c @{\hspace{1em}} c @{\hspace{1em}} c|} 
\hline
Problem & Input & \begin{tabular}[c]{@{}c@{}}Group\\ $G$\end{tabular} & \begin{tabular}[c]{@{}c@{}}Subgroup\\ $H\subseteq G$\end{tabular} & \begin{tabular}[c]{@{}c@{}}$H$-Rep\\ $R$\end{tabular} & \begin{tabular}[c]{@{}c@{}}$H$-Irrep\\ $r^{\va}_H$\end{tabular} & \begin{tabular}[c]{@{}c@{}}Output\\ $\mult{r^{\a}_H}{R}$\end{tabular} & \begin{tabular}[c]{@{}c@{}}Cost\\ $\OC(\frac{\dim(R)}{\dim(r^{\va}_H)})$\end{tabular} \\ 
\hline
Kostka & $\nu,\mu \vdash n$ & $S_n$ & $S_\mu:=\bigtimes_i S_{\mu_i}$ & $(r^\nu_{S_n})\downarrow^{S_n}_{S_\mu}$ & $\bigotimes_i r_{S_{\mu_i}}^{(\mu_i)}$ & $K_\nu^\mu$ & $\OC\big(d_\nu\big)$ \\
Littlewood & $\nu\vdash n$ and $\lm,\mu\vdash a,b$ & $S_n$ & $S_{a}\times S_{b}$ & $(r_{S_n}^\nu)\downarrow^{S_n}_{S_{a}\times S_{b}}$ & $r^\lm_{S_a} \otimes r^\mu_{S_b}$ & $c_{\lm\mu}^\nu$ & $\OC\big( \frac{d_\nu}{ d_\lm d_\mu}\big)$ \\
Plethysm & $\nu\vdash n$ and $\lm,\mu\vdash c,d$ & $S_n$ & $S_{c} \wr S_{d}$ & $(r^\nu_{S_n})\downarrow^{S_n}_{S_{c} \wr S_{d}}$ & $r^\lm_{S_c} \wr r^\mu_{S_d}$ & $a_{\lm\mu}^\nu$ & $\OC\big( \frac{d_\nu}{d_\lm^d d_\mu}\big)$ \\
\hline
Kronecker \cite{bravyi2023quantum} & $\nu,\lm,\mu \vdash n$ & $S_n\times S_n$ & $S_n$ & $(r^{\lm}_{S_n}\otimes r^{\mu}_{S_n} )\downarrow^{S_n\times S_n}_{S_n}$ & $r_{S_n}^\nu$ & $g_{\lambda \mu\nu} $ & $\OC\big(\frac{d_\lm d_\mu}{d_\nu}\big)$ \\
\hline
\end{tabular}
\label{Tab:summary}
\end{table*}

\paragraph{Introduction.--}Is there a computational problem that admits efficient quantum algorithms but no efficient randomized one? 
Quantum computing is driven by the hypothesis that the answer is yes, to a large extent because -- \textit{``...our best description of the physical world is quantum-mechanical and quantum mechanics seems hard to simulate by classical computers...''} \cite{Feynman1982}.
This intuition often comes with limited backing for specific computational problems and supporting it rigorously faces significant challenges in computational complexity theory. It would be ideal to prove that a {decision problem} admits efficient quantum algorithms but no randomized one, but doing that unconditionally entails showing that $\BPP \subset \BQP$ and in consequence separating $\BPP$ and $\PSPACE$, a major open problem in computational complexity theory. Here we resort to the next best thing: find a problem that could be outside $\BPP$; find an efficient quantum algorithm for it and study what makes it intractable classically. An excellent execution of this three-step process is Shor's algorithm for factoring and discrete logarithm \cite{shor1994algorithms}: problems that are not known to be either $\NP$-hard or $\BQP$-complete, but the evidence for their hardness is that we use factoring and discrete logarithm for public-key cryptography \cite{Diffie76, RSA78}. Here we give quantum algorithms for a class of problems in representation theory that could yield computational speedups, if our hardness conjectures could be similarly supported.

Let $G$ be a finite group. A unitary representation $R$ of $G$ on a finite-dimensional vector space $\HC=\mbb{C}^d$ is a function $R$ from $G$ to the group of unitary matrices on $\HC$ such that $R(g)R(h)=R(gh)$. Maschke's theorem asserts that for any $G$-representation $R$ there exists a unitary $U$ that simultaneously block-diagonalizes all $R(g)$ as:
\begin{equation}
UR(g)U^\dag = \bigoplus_{\va} I_{\mult{r^\va_G}{R}} \otimes r^{\va}_G(g),\,
\label{eq:block-diagonalization}
\end{equation}
where $r^\a_G$ is an irreducible $G$-representation (\emph{irrep}) labeled by $\va$.  The number of times each irrep $r^\va_G$ appears in $R$ is the \emph{multiplicity}  $\mult{r^\va_G}{R}$.
We denote $d_{R} = \dim(R)$ and $d_{\va} = \dim ( r^{\va}_G)$ the dimension of the vector spaces on which $R$ and $r^{\va}_G$ act on, use $R\downarrow^{G}_H$ for the restriction of a $G$-representation from $G$ to $H\subseteq G$, $\times$ for the direct product and $\wr$ for the wreath product. Let $S_n$ be the symmetric group on $n$ elements. Its irreps are labeled by integer partitions $\lm\vdash n$. Our algorithms compute:

\begin{enumerate}
\item \textbf{Kostka Numbers:} the multiplicity of  the trivial $S_\mu$-irrep in the restriction of an $S_n$-irrep: 
\begin{align}\label{eq_kostka}
    K^\mu_\nu = \mult{r^\triv_{S_\mu}}{r^{\nu}_{S_n} \downarrow^{S_n}_{S_\mu}},
\end{align}
where $\nu,\mu \vdash n$ and $S_\mu = \bigtimes_i S_{\mu_i}$ is a Young subgroup. $K^\mu_\nu$ also counts the number of semistandard Young tableaux of shape $\nu$ and content $\mu$. They play role in quantum tomography (see for example Ref.~\cite{Haah17}). 

\item  \textbf{Littlewood-Richardson Coefficients:} the multiplicity of $S_a \times S_b$-irreps in an $S_n$-irrep for $a+b = n$: 
\begin{align}\label{eq_lr}
    c^{\nu}_{\lm \mu} &= \mult{r^{\lm}_{S_a} \otimes r^{\mu}_{S_b}}{r^\nu_{S_n}\downarrow^{S_n}_{S_a \times S_b}}.
\end{align}
By Schur-Weyl duality \cite{fulton1991representation, Christandl05} this is also a multiplicity of $\SU(d)$-irreps in the tensor product of two other $\SU(d)$-irreps. The coefficient plays a role in angular momentum recoupling theory (for $d=2$, see for example Refs.~\cite{Wigner31, Weyl50, Pauncz79}) or nuclear and particle physics (for example for $d=3$, the ``eight-fold'' way  \cite{Gellman61}).

\item \textbf{Plethysm Coefficients:} the multiplicity of irreps of the wreath product subgroup $S_c \wr S_d\subset S_n$ for $cd=n$ in an $S_n$-irrep: 
\begin{align}\label{eq_pleth}
    a^\nu_{\lm\mu} &= \mult{r^{\lm \mu}_{S_c \wr S_d}}{r^{\nu}_{S_n} \downarrow^{S_n}_{S_c \wr S_d}}\,,
\end{align}
where $\lm,\mu,\nu \vdash c,d,n$. It is also the multiplicity of $\SU(d)$-irreps in the {composition} of a pair of them. They appear in the quantum marginal problem for distinguishable particles~\cite[Sec.~7]{Christandl14}. 
\end{enumerate}

Our construction also subsumes the algorithm for
{Kronecker coefficients} which previously appeared in Ref.~\cite{bravyi2023quantum}. Kronecker coefficients are  multiplicities of $S_n$-irreps in the tensor products of $S_n$-irreps:
\begin{align}\label{eq_kron}
    g_{\lambda \mu \nu} &= \mult{r^\nu_{S_n}}{r^{\lm}_{S_n} \otimes r^{\mu}_{S_n}},
\end{align}
and play role in the quantum marginal problem \cite{Christandl05, Christandl07}. 

Computation of Plethysm and Kronecker coefficients is worst-case $\#\P$-hard if we treat $n$ in $S_n$ as the input size (this corresponds to the \emph{unary} input encoding) \cite{Burgisser08, Fischer2020} and it has been conjectured that the same is true for the Littlewood-Richardson and Kostka coefficients \cite[Conjecture 5.11]{panova2023computationalcomplexityalgebraiccombinatorics}~\footnote{The computation of Littlewood-Richardson coefficients (to which the Kostka number computation reduces parsimoniously) is $\#\P$-hard for \emph{binary} encoding by a reduction to Knapsack \cite{narayanan2006complexity}. Binary encoding encodes the length of each row in the input integer partition in binary; this means that for a fixed input size $n$, the order of the underlying symmetric group ranges from $S_n$ to $S_{2^n}$. Fixing the order of $S_n$, as we do here, is equivalent to using unary encoding of the input. It is open if Kostka and LR-coefficients are $\#\P$-hard with unary encoding. See~\cite{panova2023computationalcomplexityalgebraiccombinatorics}.}. While it would be surprising to find a quantum algorithm that computes the coefficients efficiently on all inputs, we give quantum algorithms that are efficient on nontrivial subsets of  inputs and analyze their performance.

\paragraph{Quantum Algorithms.--}

Let $G$ be a finite group and $R$ be a $G$-representation on $\mbb{C}[G]$ with orthonormal basis $\{\ket{g}\}_{g\in G}$ acting as:
\begin{align}
    R(h)\ket{g} &= \ket{hg}, \quad \text{for all $h,g\in G$.}
\end{align}
Such representation is called left-\emph{regular representation}. A transformation $U$ that maximally block-diagonalizes this representation as in Eq.~\eqref{eq:block-diagonalization} is the \emph{Quantum Fourier Transform} (QFT); we write $G$-QFT whenever we need to make the dependence on $G$ explicit.
It is defined as: 
\begin{align}
  \langle \a,i,j|U|g \rangle &= \sqrt{\frac{d_\a}{|G|}}\braket{i}{r_G^\a(g)|j},
  \label{Eq:QFT_matrix_elements}
\end{align}
where $r_G^\a(g)$ is a $G$-irrep labeled by $\a$ and $\braket{i}{r_G^\a(g)|j}$ for $i,j \in [d_\alpha]$ is its $i,j$-th matrix element. Efficient quantum algorithms for $G$-QFT are known for $G=S_n$~\cite{beals1997quantum} and some other finite groups 
~\cite{Moore07}. We use QFT over the Young subgroup $S_\mu = \bigtimes_i S_{\mu_i}$ and the wreath product group $S_a \wr S_b$ from Ref.~\cite{Moore07} (see Supp. Mat.). 

Our algorithm also relies on the efficient preparation of states of the form:
\begin{align}
\rho^{\a}_G &= 
U^\dag \left( \ket{\a}\bra{\a} \otimes \frac{I_{d_\a}}{d_\a} \otimes \frac{I_{d_\a}}{d_\a} \right) U = \frac{\Pi^G_{\a}}{ d_\a^2},
\label{Eq:InputState}
\end{align}
where $I_{d_\a}$ is the $d_\a \times d_\a$ identity matrix and $\Pi^G_{\a}$ is the projector onto the maximally isotypic subspace $\a$ in the regular representation of $G$; where the maximally isotypic subspace of a given type is the subspace of $\HC$ containing all irreps of such type. Eq.~\eqref{Eq:InputState} shows that an efficient $G$-QFT is sufficient, but possibly not necessary, to prepare this state. Equipped with this, we prove:

\begin{theorem}\label{theorem_quantum_algo}
Let $H\subseteq G$. Let $\mathcal{H}$ be the size of a quantum circuit for $H$-QFT, $\mathcal{R}$ be the size of the circuit for preparing $\rho^{\a}_G$ in Eq.~\eqref{Eq:InputState} and $\mathcal{D}$ be the size of a circuit that implements a controlled action of $G$-regular representation. Then there is an $\OC((\mathcal{H} + \mathcal{R} + \mathcal{D}) (d_\alpha/d_\beta)^2)$ quantum algorithm for the following problem: given a pair of irrep labels $(\alpha, \beta)$ as an input, compute  $\mult{r^\b_H}{r^\a_{G}\downarrow_H^{G}}$ (with constant probability of success) \footnote{The irreps $\alpha, \beta$ can be labeled with $O(\log|G|)$-many bits each, so that the input size $n$ relates non-trivially to the order of the group $G$. Since there are at most $|G|$ distinct irreps, such labelling always exists.}.
\end{theorem}

\begin{proof}
The algorithm is a variation on Harrow's generalized phase estimation \cite{harrow2005applications} shown in Fig.~\ref{fig_1}.
Let $G$ be a finite group and let $H\subseteq G$. 
The algorithm uses two registers: a {target register} $\mathcal{T}$ and a {control register} $\mathcal{C}$.
It inputs $\rho^\a_{G}$ from Eq.~\eqref{Eq:InputState} to the target register and a uniform superposition over $H \subseteq G$ to the control, prepared by applying an inverse $H$-QFT to the state $\ket{\text{triv},1,1}$, where $\text{triv}$ labels the trivial irrep of $H$. The algorithm then applies a controlled action of the $G$-regular representation $R$ to $\rho^{\a}_{G}$. Since $R$ is controlled by a register that contains a superposition over the subgroup $H$, this acts as a restricted representation $R\downarrow_H^{G}$ and decomposes into irreducible representations of $H$. The algorithm concludes by measuring the $H$-irrep label on the control register. All of this can be described by: 
\begin{equation}
\begin{aligned}
    &E_\b = \frac{1}{\sqrt{|H|}}   \sum_{h\in H} \Bigg[ (\ketbra{\b}{\b}\otimes I_{d_\b}\otimes I_{d_\b})V\ket{h} \Bigg]_{\mathcal{C}} \otimes \Bigg[R(h) \Bigg]_\mathcal{T},
\end{aligned}
\end{equation}
where $V$ is the $H$-QFT and $\beta$ labels inequivalent irreps of $H$. A measurement outcome $\beta$ appears with probability:
\begin{equation}
\begin{aligned}
    p_\a(\b) &= \text{Tr}\left[E_\b  \rho^\alpha_{G} E_\b^\dag \right] 
    = \mult{r^\b_{H}}{r^\a_{G}\downarrow^{G}_H} \frac{d_\b}{d_\a} . 
    \label{eq_prob}
\end{aligned}
\end{equation}
Because $\mult{r^\beta_{H}}{r^\a_{G} \downarrow^{G}_H}$ is an integer, we can get its exact value in $\OC((d_\a/d_\b)^2)$-many shots with high probability by the Chernoff-Hoeffding bound (see Supp. Mat.). The algorithm computes $\mult{r^\b_{H}}{r^\a_{G} \downarrow^{G}_H}$ in $\OC((\mathcal{H} + \mathcal{R} +\mathcal{D})) (d_\a/d_\b)^2)$ time.  
\end{proof}
Tab.~\ref{Tab:summary} summarizes different choices of $H$ and $G$ where the quantum algorithm in Thm.~\ref{theorem_quantum_algo} computes the coefficients in Eqs.~\eqref{eq_kostka} to \eqref{eq_kron}. 
Note that nearly all cases considered have the same ambient group $G=S_n$. Since $\log |S_n| = \log n! \approx n \log n \in \widetilde{\OC}(n)$, we parametrize the complexity of the algorithms below by $n$. We parametrize the complexity of the problem in terms of dimensions of $S_n$-irreducible representations which can be computed in $\OC(n^2)$ using the hook-length formula (see the Supp. Mat.). Following results Ref.~\cite{kawano2016quantum}, we also assume that the $S_n$-QFT can be implemented in gate complexity $\OC(n^4)$.

\begin{figure}[h]
\centering
\includegraphics[width=\columnwidth]{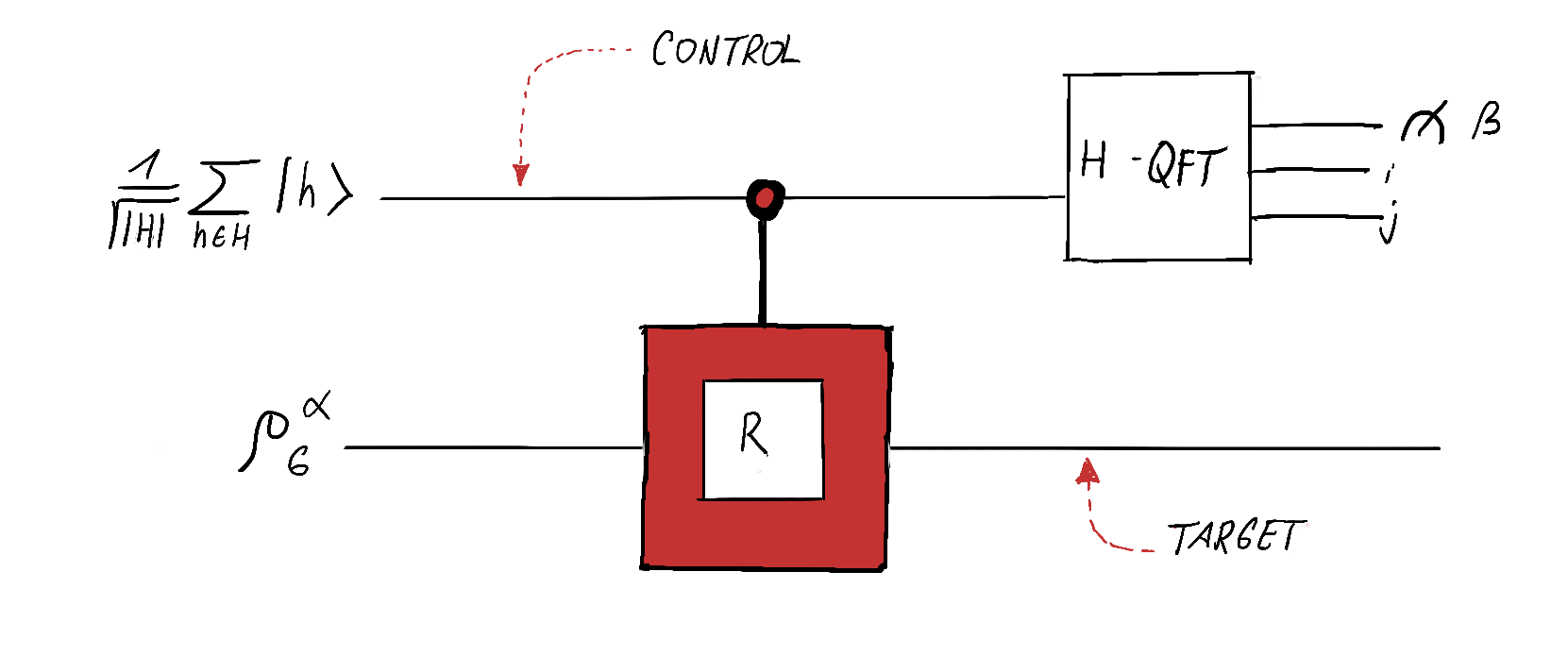}
\caption{The quantum algorithm in Thm.~\ref{theorem_quantum_algo}. The algorithm initializes the state from Eq.~\eqref{Eq:InputState}, applies a $G$-regular representation controlled by a uniform superposition over $H$ and weak-Fourier samples the control register.
}
\label{fig_1}
\end{figure}

\paragraph{Kostka Numbers.---}
 The Kostka number $K_\nu^\mu$  is the multiplicity of the trivial irreducible representation in the restriction of an $S_n$-irrep $r^\nu_{S_n}$ to $S_\mu$, $K^\mu_\nu = \mult{r_{S_\mu}^{\triv}}{r^\nu_{S_n} \downarrow_{S_\mu}^{S_n}}$.  A quantum algorithm for the Kostka numbers fixes $H = S_\mu \subseteq G=S_n$ for an integer partition $\mu \vdash n$, where $S_\mu$ is the Young's subgroup
$S_\mu := \bigtimes_i S_{\mu_i}$. An efficient $S_\mu$-QFT can be setup from a tensor product of $S_{\mu_i}$-QFTs for each factor in $\mu \vdash n$ and its complexity is comparable to $S_n$-QFT. Theorem \ref{theorem_quantum_algo} then gives the label of the trivial irreducible representation of $S_\mu$ with probability (c.f. Eq.~\eqref{eq_prob}):
 \begin{align}
     p_{\nu,\mu}(\triv) &=  K^\mu_\nu / d_\nu.
     \label{Eq:kostka}
 \end{align}
 This suffices to compute the Kostka number in $\OC(n^4 d_\nu^2)$ Computing Kostka numbers $K^\mu_\nu$ on an input $\mu, \nu$ in \emph{unary} is conjectured to be $\#\P$-complete in \cite[Conjecture 5.11]{panova2023computationalcomplexityalgebraiccombinatorics} and one could hope that restricting $d_\nu \leq \poly(n)$ remains classically hard. We show that this is not the case: 
 \begin{theorem}
     There is a deterministic algorithm that on an input $\mu,\nu \vdash n$ in \emph{unary}, such that $d_\nu \leq \poly(n)$,  computes the Kostka number $K_\nu^\mu$ in polynomial time in $n$.
     \label{Thm:Kostkas}
 \end{theorem} 
 \noindent See~the Supp. Mat. for a proof.


 \paragraph{Littlewood-Richardson Coefficients.--}
 The quantum algorithm for Littlewood-Richardson (LR) coefficients fixes $H = S_a \times S_b \subseteq G=S_n$  for $a + b = n$. The irreducible representations of $S_a \times S_b$ are a tensor product of irreducibles $r_{S_a}^\lm \otimes r_{S_b}^\mu$, labeled by pairs $\a = (\lm,\mu)$ with $\lm,\mu \vdash a,b$. We have $c^\nu_{\lm \mu } = \mult{r^\lm_{S_a} \otimes r^\mu_{S_b}}{r^{\nu}_{S_n} \downarrow^{S_n}_{S_a \times S_b}}$, and from Thm.~\ref{theorem_quantum_algo}, the algorithm samples from:
 \begin{align}
     p_{\nu}(\lm,\mu) &=  c^{\nu}_{\lm \mu}\frac{d_\lm d_\mu}{d_\nu}.
 \end{align}
 This computes the LR coefficient in $\OC\left(n^4 \left(\frac{d_\nu}{d_\lambda d_\mu}\right)^2 \right)$ gate complexity.
 
LR coefficients count combinatorial objects called the LR tableaux. A skew shape $\nu/\lm$ is a diagram where $\lm$ is subtracted from $\nu$ 
 and a skew semistandard Young Tableau is a skew shape filled with integers such that the entries are non-decreasing along every row and strictly increasing along every column. The LR tableau is a skew Young tableau of shape $\nu/\lm$ whose row-reading evaluates to $\mu$. The LR coefficient $c^\nu_{\lambda \mu}$ is the number of such tableaux; this fact is called the LR rule. See~\cite{Howe12, James84, sagan2001symmetric} and the Supp.~Mat. for details (esp. definition of `row reading'). The quantum algorithm in Thm.~\ref{theorem_quantum_algo} is only efficient for $c^\nu_{\lm \mu} \leq \poly(n)$ in which case it counts only up to $\poly(n)$ many LR tableaux. 
In analogy to Thm.~\ref{Thm:Kostkas}, we expect that there is a classical algorithm that efficiently computes LR coefficients under this constraint:
\begin{hypothesis}
 There is a polynomial-time randomized algorithm for computing  LR coefficients $c^\nu_{\lm \mu}$ for $\nu \vdash n$, $\lm \vdash a$ and $\mu \vdash b$ for $a+b = n$ and $d_\nu /d_\lm d_\mu \leq \poly(n)$.
     \label{Conj:LR}
\end{hypothesis} See also \cite{panova2025polynomialtimeclassicalversus}. If true, this rules out superpolynomial speedups for computing LR coefficients by Thm.~\ref{theorem_quantum_algo}. \footnote{We make a distinction between a hypothesis and a conjecture. By a hypothesis, we mean claim that we strongly believe in but cannot prove. A conjecture labels a claim that we cannot prove and have no any strong beliefs about.}

\textit{Plethysm Coefficients.---}
  The algorithm for Plethysm coefficients sets $H = S_c \wr S_d \subseteq G=S_n$ for $cd = n$.  An efficient implementation of the QFT over this group is known only for $|S_c| \leq \poly(n)$ \cite{Moore07}.  
  Under this restriction, the quantum algorithm in Thm.~\ref{theorem_quantum_algo} samples efficiently from a distribution:
 \begin{align}
     p_\nu(\lm,\mu) &= a^\nu_{\lm\mu} \frac{d_\lambda d_\mu^d }{d_\nu},
 \end{align}
 where $a^\nu_{\lm\mu}=\mult{r^{\lm,\mu}_{S_c \wr S_d}}{r^{\nu}_{S_n}\downarrow^{S_n}_{S_c \wr S_d}}$ is the  Plethysm coefficient. It is not known if Plethysm coefficients count a well defined set of combinatorial objects (Prob.~9 on Stanley's list \cite{Stanley99}), and the simulation strategy of Thm.~\ref{Thm:Kostkas} doesn't straightforwardly apply \cite{Ikenmeyer17, Fischer2020, pak2015complexity}. The classical complexity of computing the Plethysm coefficients was recently studied in Ref.~\cite{Fischer2020} where it was shown that the problem is in $\P$ if the number of parts $\ell(\lambda)$ is constant and $\mu = \triv$, but becomes $\#\P$-hard even for fixed $d = 3$ and $\ell(\lambda)$ unrestricted. While this suggests that the complexity of the problem is  controlled by the parameter $c$, its complexity for $|S_c| = \poly(n)$  ($c \sim \log(n)$) remains, to our best knowledge, open. We thus state:
 \begin{conjecture}[See also \cite{panova2025polynomialtimeclassicalversus}]
     There is no polynomial time randomized algorithm that (with high probability) computes $a^\nu_{\lambda \mu}$ for $\lambda \vdash c$ and $\mu \vdash d$, such that $cd = n$, $c \leq \poly(\log(n))$ and $d_\lambda d_\mu^d / d_\nu \geq 1/\poly(n)$.
     \label{Conj:Plethysms}
 \end{conjecture}
 If true, the quantum algorithm for Plethysm coefficients gives a superpolynomial quantum speedup. The fact that an analog of Thm.~\ref{Thm:Kostkas} does not apply here supports this, but the constraints on the conjectured computationally hard regime for the problem are very stringent and can make the restricted computational problem significantly easier. We have little intuition whether the conjecture is true or false, besides the obstructions proved in \cite{panova2025polynomialtimeclassicalversus} while we were finishing this manuscript.

\textit{Kronecker Coefficients.---}  Given $\lm,\mu,\nu\vdash n$, the Kronecker coefficient $g_{\lm \mu \nu}$ is defined as $g_{\lm \mu \nu} = \mult{r^\lm_{S_n} \otimes r^\mu_{S_n}}{r^\nu_{S_n}}$. The algorithm in Ref.~\cite{bravyi2022quantum} is a special case of Thm.~\ref{theorem_quantum_algo}, for the choice $H=S_n \subseteq G=S_n\times S_n$. It samples from the output distribution: 
\begin{align}
    p_{\lm \mu}(\lm) &= \frac{d_\nu}{d_\lm d_\mu } g_{\lm \mu \nu}\,,
\end{align} and computes $g_{\lm \mu \nu}$ (with high constant probability) in time $\OC \left(n^4  (d_\lm d_\mu / d_\nu)^2\right)$. Similarly to Plethysms, it remains a major open question if the Kronecker coefficients count some simple set of combinatorial objects (Stanley's Prob.~10 \cite{Stanley99}), which again precludes straightforward application of the simulation strategy of Thm.~\ref{Thm:Kostkas}. The classical complexity of computing the Kronecker coefficients was studied in Ref.~\cite{pak2015complexity} which identified the extensive number of parts as a potential obstruction to their efficient computation. They show that $g_{\lm \mu \nu}$ can be computed in time: 
\begin{align}
    \OC \Big(L \log(N) \Big) + \Big(L \log (M) \Big)^{\OC(L^3 \log L)},
    \label{Eq:PakComplexity}
\end{align}
where $L=\max \Big(\ell(\lm),\ell(\mu),\ell(\nu) \Big)$, $N = \max(\lm_1,\mu_1,\nu_1)$ is the largest length of the first row and $M = \max(\lm_2,\mu_2,\nu_2)$ is the largest length of the second row among these partitions \cite[Thm.~3.5]{pak2015complexity}. 
Even for fixed $L$, the algorithms of Refs.~\cite{Christandl12, Baldoni18} compute the Kronecker coefficient in polynomial time but become inefficient when the length of the input partitions varies with input size. The algorithm in Ref.~\cite{mishna2022estimating} computes Kronecker coefficient by reduction to vector partition function problem. All of these algorithms are based on Barvinok's algorithm for counting the lattice points in convex polytopes (Refs.~\cite{Loera2004EffectiveLP, DeLoera2004, DeLorea2005b, DeLoera2005, Barvinok93, BarvinokPomershein99, DyerKannan97}). 

The quantum algorithm in Thm.~\ref{theorem_quantum_algo} is for all practical purposes outperformed by such algorithms on {some} inputs. 
For example, Ref.~\cite[Ex.~3.4]{mishna2022estimating} computes $g_{\lambda \mu \nu} = 391$ for $\lm = [140^2], \; \mu = [70^4],\;  \nu = [57^3,33^3,10]$ for which $d_\nu/(d_\lm d_\mu  ) \approx 1.8 \times 10^{-31}$ on which the quantum algorithm in Thm~\ref{theorem_quantum_algo} requires an order of $10^{31}$ samples. On the other hand, it runs in polynomial when $d_\mu \leq \poly(n)$ and $\lambda = \nu$ is arbitrary. We give an explicit construction of a family of such inputs in the Supp. Mat.. 

We originally took the {length obstruction} in Eq.~\eqref{Eq:PakComplexity} as a support for the conjecture of Ref.~\cite{bravyi2023quantum} that the algorithm for the Kronecker coefficients can outperform classical algorithms for triples of partitions with large number of parts. This lead us to state: 
\begin{conjecture}[Disproved in \cite{panova2025polynomialtimeclassicalversus}]
    There is no polynomial time randomized algorithm that (with high probability) computes $g_{\lm \mu \nu }$ for $\mu, 
    \nu, \lambda \vdash n$ and $d_\mu \leq \poly(n)$.
    \label{Conjecture:Kron}
\end{conjecture}

 However, Greta Panova disproved this conjecture by extending the result of Ref.~\cite{Christandl12} by showing that whenever $d_\mu \leq n^k$, there is a $\tilde{\OC}(n^{4k^2+1})$ classical algorithm for the problem \cite[Theorem 1.] {panova2025polynomialtimeclassicalversus}. A more general form of this conjecture (as discussed in \cite{bravyi2023quantum}) narrowly escapes this disproof but requires fine-tuned inputs, so we take Panova's result as a strong evidence against superpolynomial quantum speedups for this problem (see footnote \footnote{There is no polynomial time randomized algorithm that (with high probability) computes $g_{\lm \mu \nu }$ for $\mu, 
    \nu, \lambda \vdash n$ and $d_\nu/(d_\lambda d_\mu) \geq 1/\poly(n)$. The key challenge then becomes to find a nicely parametrized set of inputs such that $d_\nu, d_\lambda, d_\mu$ are all superpolynomial, their ratio is bounded by a polynomial and the partitions do not fall into one of the special cases that can be simulated easily (such as hooks). As long as $\lm, \mu$ have number of parts that scale with $n$, the quantum algorithm can do better than the bound in Eq.~\eqref{Eq:PakComplexity} as well as the state of the art classical algorithms, but the new result in \cite{panova2025polynomialtimeclassicalversus} supersedes this obstruction}).  Interestingly, the gate complexity of our algorithm is ${\OC}(n^{4} (d_\mu d_\lambda/d_\nu)^2)$ so that a ${\OC}(n^{4+2k})$ vs. $\tilde{\Omega}(n^{4k^2+1})$ polynomial speedup for $d_\lambda = d_\nu$ and $d_\mu \leq n^k$ is possible, assuming that the scaling in \cite[Theorem 1.] {panova2025polynomialtimeclassicalversus} is tight. As this asks for additional scrutiny, we pose:
    \begin{question}
    \label{ref:question}
    Is there a classical algorithm that computes $g_{\mu \nu \lambda}$ in time ${\OC}(n^{2k+4})$ whenever $d_\mu \leq n^k$?
    \end{question}
    
    See the Supp. Mat. for an explicit construction of a nontrivial family of inputs on which our algorithm is efficient, additional analysis, discussion and numerical study.

\paragraph{Induction Algorithm.--}
Frobenius's reciprocity~\cite{fulton1991representation}:
\begin{equation}
     \mult{r^\b_H}{ r^\a_G \downarrow_H^G} =\mult{r^\a_G} {r^\b_H \uparrow_H^G}\,,
\end{equation}
suggests an alternative to Thm.~\ref{theorem_quantum_algo} that computes multiplicities by Fourier sampling over induced representations: 

\begin{theorem}\label{theorem_quantum_algo_ind}
Let $H\subseteq G$. If there is a quantum circuits for $G$-QFT with $\mathcal{G}$ many gates, a quantum circuit with $\mathcal{R}$ many gates for preparing $\rho^{\b}_H$ in Eq.~\eqref{Eq:InputState}, then there is a $\OC((\mathcal{G} + \mathcal{R}) (d_\b |G| /d_{\a} |H|)^2)$ quantum algorithm that, given a pair of irrep labels $(\alpha, \beta)$ as an input, computes (with high probability) $\mult{r^\a_G} {r^\b_H \uparrow_H^G}$. 
\end{theorem}

\begin{proof}
See Fig.~\ref{fig_2:induction}.
    The algorithm uses two registers, an $H$-regular rep register $\mathcal{S}=\mbb{C}[H]$ and a coset register $\mathcal{V}=\mbb{C}[G/H]$. It  prepares $\rho^\beta_H$ from Eq.~\eqref{Eq:InputState} in $\mc{S}$ and $|H|\id/|G| $ in $\mathcal{V}$ and applies an `embedding' unitary $U_E$ to $\mc{S}\otimes \mc{V}$: 
$    
        U_E \Big(\ket{t}_\mathcal{V} \otimes \ket{h}_\mathcal{S} \Big) = \ket{th},
$
    where $t$ is a representative of the coset of $H$ in $G$ and $th$ is an element of $G$. This unitary can be implemented efficiently as long as the multiplication of the group elements in $G$ is efficient an we assume that it's complexity is negligible compared to both $\mathcal{G}$ and $\mathcal{R}$. The algorithm  applies a $G$-QFT to $\mc{S}\otimes \mc{V}$ and measures the irrep label. The outcome $\alpha$  appears with probability: 
    \begin{align}
        q_\beta(\alpha) &= \frac{|H|d_\alpha}{|G| d_\beta} \mult{r^\alpha_G}{r^\beta_H \uparrow_H^G}.
        \label{Eq:Induced_Distribution}
    \end{align}
     Following the analysis in Thm.~\ref{theorem_quantum_algo}, we can find $ \mult{r^\alpha_G}{r^\beta_H \uparrow_H^G}$ in $\OC((\mathcal{G} + \mathcal{R}) (d_\b |G| /d_{\a} |H|)^2)$ with high probability. See Supp.~Mat. for details.
\end{proof}

\begin{figure}[h]
\centering
\includegraphics[width=\columnwidth]{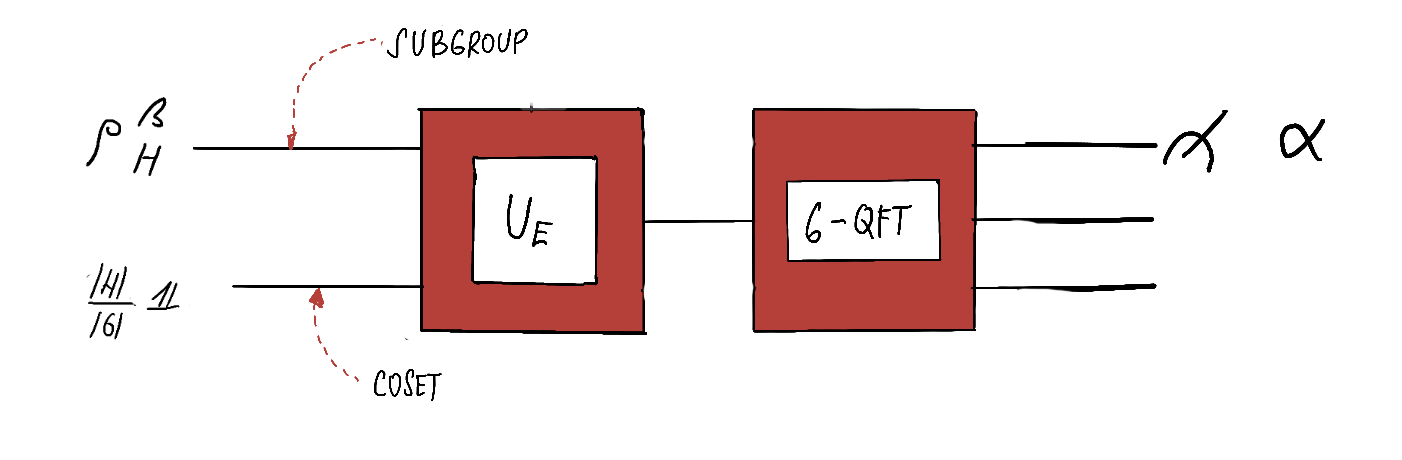}
\caption{The quantum algorithm in Thm.~\ref{theorem_quantum_algo_ind}. The algorithm applies the embedding unitary $U_E$ to the subgroup and coset register and weak Fourier samples.}
\label{fig_2:induction}
\end{figure}

The sampling cost of the restriction algorithm in Thm.~\ref{theorem_quantum_algo} is $ d_\a / d_\b $, compared to $(d_\b|G|)/(d_\a|H|)$ in Thm.~\ref{theorem_quantum_algo_ind} and the restriction is preferable whenever
$ \label{eq_cost_compare}
    (d_\a/d_\b)^2 < |G|/|H|$. For $\beta = \triv$, Eq.~\eqref{Eq:Induced_Distribution} is the distribution of Weak Fourier sampling from the hidden subgroup state, appearing in the context of the Hidden Subgroup Problem (HSP). The computational task in HSP is of course different from  computing the multiplicity, but the connection suggests generalizations of HSP to $\b \neq \triv$ we explore in future work.

\paragraph{Discussion.--}

We gave quantum algorithms for Kostka, Littlewood-Richardson, Kronecker and Plethysm coefficients, efficient on a nontrivial subset of inputs. 
We gave a polynomial-time classical algorithm for computing of the Kostka numbers $K^\mu_\nu$ whenever $d_\nu \leq \poly(n)$, ruling out a superpolynomial quantum speedup for this problem and conjectured the existence of a similar classical algorithm for the Littlewood-Richardson coefficients. 
We explained why such classical algorithm does not straightforwardly extend to the Plethysm and Kronecker coefficients and conjectured that our quantum algorithms efficiently solve problems that are hard for classical algorithms. Going against our thesis is a result in Ref.~\cite{jordan2008fast} in which Jordan presented a quantum algorithm for approximating $S_n$ characters to a similar precision as our Thm.~\ref{theorem_quantum_algo} and gave an efficient classical algorithm based on the Roichman's rule \cite{Roichman1999}. Our Conjecture \ref{Conjecture:Kron} followed a similar fate and was disproved by an algorithm in \cite[Theorem 1.]{panova2025polynomialtimeclassicalversus}. The runtime of this algorithm is $\tilde{\OC} (n^{4k^2 +1})$ which, if tight, suggests a possible ${\OC}(n^{4+2k})$ vs $\tilde{\Omega}(n^{4k^2+1})$ speedup on inputs with $d_\nu = d_\lambda$ and $d_\mu \leq n^k$. It is thus important to address Question \ref{ref:question} and understand if the large polynomial gap for computing Kronecker coefficient closes. Conjecture~\ref{Conj:Plethysms} was put under similar scrutiny: there are significant limitations to regimes where the algorithm can offer a superpolynomial speedup, but large polynomial speedups may be still possible. 

We note that, Ref.~\cite{crichigno2024quantum} showed that some of the multiplicity coefficients considered here appear naturally in spin chains. It remains an interesting question if this approach could provide more insight into this class of quantum algorithms.

\paragraph{Acknowledgments.--} We would like to thank David Gosset, Sergey Bravyi, Louis Schatzki, Greta Panova, Marco Cerezo, Marcos Crichigno, Samuel Slezak, Diego Garcia-Mart\'in, Frederic Sauvage and Dmitry Grinko for discussions. Last but not least, we would like to thank an anonymous QIP referee who helped us to significantly improve the presentation of this paper, especially regarding the presentation of the correctness of the $S_n$-QFT - the previous version contained many mistakes. We had decided to rewrite that appendix altogether and if you ever get a chance to read this, we are seriously grateful for the detailed feedback and sorry for not updating our manuscript earlier. M.L.
was supported by the Center for Nonlinear Studies at
LANL and by the Laboratory Directed Research and
Development (LDRD) program of LANL under project
number 20230049DR.

\bibliography{cmd/biblio}

\clearpage
\newpage
\onecolumngrid

\setcounter{lemma}{0}
\setcounter{proposition}{0}
\setcounter{theorem}{0}

\section{Supplemental Material}

\subsection{Representation Theory: Quick Reference}
We provide a quick reference for most of the concepts from representation theory that we use in the article. We refer the reader to the books of Fulton and Harris~\cite{fulton1991representation}, Goodman and Wallach~\cite{goodman2009symmetry} or Serre~\cite{serre1977linear} and Sagan's book~\cite{sagan2001symmetric} and James's lecture notes~\cite{James84} for Representation Theory of the Symmetric Group. A \textbf{group} is a set $G$ together with an \textit{associative} operation $\circ$ under which a given element $e\in G$ is the identity $e\circ g = g\circ e = g$, and for each $g\in G$ there is an inverse $g^{-1}\in G$ such that $g\circ g^{-1}=e$. Some examples are:
\begin{itemize}
   \item  The \textit{symmetric group} is the set of all permutations of $n$ elements, e.g.
\begin{equation}
   S_3 = \{ (), (1,2),(1,3),(2,3),(1,2,3),(1,3,2) \}\,.
\end{equation}
\item The \textit{special unitary group} of degree $d$, $\SU(d)$, is the set of all $d\times d$ matrices with unit determinant. It is a subgroup of $\GL(d)$, the space of invertible linear operators on $\mbb{C}^d$.
\end{itemize}

The action $\circ$ of groups can be extended to vector spaces $\HC$, if we map each group element $g$ to a linear operator $R(g)$ on $V$, provided we ask $R$ to preserve the group structure $R(g) R(h) = R(g\circ h)$ for all $g,h$. A representation of $G$ on $\HC$ is a map:
\begin{equation}
R:G \xrightarrow{} \GL(\HC),
\end{equation}
that satisfies $R(g) R(h) = R(g\circ h)$, called the \textit{group homomorphism} condition. We will oftentimes also call $\HC$ the representation of $G$, being implicit about $R$. More generally, the vector spaces $\HC$ that allow for the action of $G$ are called $G$-modules.

Let $R$ be a representation for a group $G$ , with $\HC$ the corresponding $G$-module.
Assuming $G$ to be finite or compact, Maschke's theorem states that any representation can be expressed in terms of its decomposition into \textit{irreducible representations} (or irreps) 
\begin{equation}
R \cong \bigoplus_{\lm \in \hat{G}} I_{{\mult{\lm}{R}}} \otimes r^\lm_G\,,\quad \text{and}\quad   \HC \cong \bigoplus_{\lm \in \hat{G}} \mbb{C}^{\mult{\lm}{R}} \otimes \mc{H}^\lm\,.
\end{equation}
Here, $\hat{G}$ is the set of all inequivalent irreps of $G$, $\lm$ labels a given irrep, $\HC^\lm\subset \HC$ are irreducible modules and $r^\lm_G:G\arr {\rm End}(\HC^\lm)$ the irreducible group homomorphisms. We will call the subspace collecting all irreps of same type $\lm$, $\HC_\lm= I_{m_\lm} \otimes \HC^\lm \subset \HC$, the (maximal) \textit{isotypic subspace} of type $\lm$ in $\HC$.

In order to introduce the Representation Theory of $S_n$ and $\SU(d)$, we have to introduce integer partitions, Young diagrams and Young tableaux.
A \text{partition} of a positive integer \( n \) is a way of writing \( n \) as a sum of positive integers in weakly decreasing order, where the order of the addends does not matter. Formally, a partition is represented as a tuple:
\begin{equation}
   \lambda = (\lambda_1,\lambda_2,\cdots,\lambda_d) \vdash n \,\quad \text{ with } \lambda_1 \geq \lambda_2 \geq \cdots \geq \lambda_d > 0 \text{ and } \sum_{i=1}^{d} \lambda_i = |\lm| = n.
\end{equation}
Here, $l(\lm)=d\leq n$ is the number of parts in $\lm$, also called its \textit{length}. For example, the partitions of $n = 4$ are $4$, $3+1$, $2+2$, $2+1+1$ and $1+1+1+1$. Their lengths are $1,2,2,3$ and $4$, respectively.
We will use the shorthand $(\cdots,k^{m_k},\cdots)$, where $m_k$ is the multiplicity of $k$ in $\lm$, to denote partitions more compactly. For example, $(1,1,1,1)=(1^4)\vdash 4$ and $(4,3^2,1^5)\vdash 15$. 
The set of all partitions of $n$ is denoted $\PC_n$ and the subset of partitions with length at most $d$ parts is denoted $\PC_{n,d}\subset \PC_n$. We note that sometimes it will be convenient to pad partitions with $0$'s. In any case, $\lm$ and $\lm+(0^{m_0})$ (for any positive integer $m_0$) denote the same partition.

\medskip
Given a partition $\lm \vdash n$ (the symbol $\vdash$ is used to denote partitions), its associated Young diagram consists of a diagram with $k$ rows and $\lm_k$ left-justified boxes on the $k$-th row. For example, the Young diagram for $\lm=(3,1)$ is  \ydiagram{3,1}. A tableau is a filling of a Young diagram with integers. \\

A \textit{Standard Young Tableau} (SYT) is a filling of a Young diagram with positive integers in a strictly increasing manner, both row and column-wise. 
Examples of SYTs with shape $\lm=(3,1)$ are
\begin{equation}\label{eq_SYT}
\ytableausetup{boxsize=normal,aligntableaux=center}
\begin{ytableau}
   1 & 2 & 3 \\
   4
\end{ytableau}
\hspace{.5cm} 
and
\hspace{.5cm} 
\begin{ytableau}
   1 & 2 & 4 \\
   3
\end{ytableau}\,\quad.
\end{equation}

A \textit{Semi-Standard Young Tableau} (SSYT) is a filling of a Young diagram with non-decreasing positive integers across each row and strictly increasing down each column.
Examples of SSYTs with shape $\lm=(3,1)$ are
\begin{equation}\label{eq_SSYT}
    \ytableausetup{boxsize=normal,aligntableaux=center}
    \begin{ytableau}
       1 & 1 & 1 \\
       2
    \end{ytableau}
    \hspace{1cm}
    \begin{ytableau}
       1 & 1 & 2 \\
       2
    \end{ytableau}
    \hspace{1cm} 
    \begin{ytableau}
       1 & 2 & 2 \\
       2
    \end{ytableau}\quad.
\end{equation}
Let ${\rm SYT}(\lm)$ and ${\rm SSYT}(\lm)$ be, respectively, the sets of SYTs and SSYTs of shape $\lm$. The \textit{content} of a skew Tableau is a sequence counting the number of appearances of each integer. For example, the previous SSYT have content $\mu=(3,1),(2,2)$ and $(1,3)$, respectively. Integer partitions $\lm\vdash n$ label $S_n$ and $\SU(d)$ irreps,
\begin{align}\label{eq_specht}
r^{\lm}_{S_n}:S_n\arr \U(\mc{S}^\lm) \quad \text{and} \quad r^{\lm}_{\SU(d)}:\SU(d)\arr \U(\mc{W}^\lm)\,,
\end{align}
with $\mc{S}^\lm$ and $\mc{W}^\lm$ the irreducible modules (called Specht and Weight modules, respectively), while standard and semistandard Young tableaux of shape $\lm$ label orthonormal bases for such irreducible modules:
\begin{align}\label{eq_weight_mod}
\mc{S}^\lm = \spn \Big\{ \text{ $\ket{T}$\,,\, $T\in{\rm SYT}(\lm)$} \Big\}\,\quad \text{and}\quad    \mc{W}^\lm = \spn \Big\{ \text{ $\ket{q}$\,,\, $q\in{\rm SSYT}(\lm)$} \Big\}\,.
\end{align}
In the case $\lm=(1,0^{d-1})$, the $\SU(d)$ irrep $r^\lm_{\SU(d)}$ is called the \textit{standard representation} and corresponds to the representation of $d\times d$ special unitary matrices \textit{as themselves}. In the case of $S_n$, there is no standard representation since the group is not a matrix group. Instead, we can mention the \textit{defining representation} of any permutation $\sg \in S_n$, consisting of its action on $\mbb{C}^n$ by permuting the canonical basis $\sg\cdot \ket{i} = \ket{\sg^{-1}(i)}$, or its \textit{regular representation}, where the module is the group algebra $\mbb{C}[S_n]$ and the action is the one inherited from the group operation $\sg\cdot\ket{\pi}=\ket{\sg\cdot \pi}$. None of these is irreducible.

A \textit{Yamanouchi word} (also known as a \textit{ballot word}) of length $n$ in the alphabet $[k]\equiv\{1,2,\ldots,k\}$ is a word

\begin{equation}\label{eq_Yamanouchi}
w = w_1w_2\cdots w_n \quad \text{with each }w_i\in [k]
\end{equation}
with the property that for every prefix $w^j=w_1w_2\cdots w_j$, for any $j \in [n]$, and for each $i\in[k]$, the number of occurrences of the letter $i$ is at least as many as that of $i+1$. Let the set of all $n$ $k$-strings be $\mc{X}_{n,k}$, where of course $|\mc{X}_{n,k}|=[k]^n$, and the Yamanouchi Words (YW) constitute a strict subset. The Yamanouchi condition can be compactly stated by introducing the \textit{type} of any $x\in \mc{X}$ as 
\begin{equation}
t(x)=(t_1(x),\cdots,t_k(x))
\end{equation}
where $t_i(x)$ is the number of $i \in [k]$ in $x$. Of course, $\sum_i t_i(x)=n$. For example, for $x=11212\in \mc{X}_{5,2}$ we have $t(11212)=(3,2)$. The Yamanouchi condition is
\begin{equation}
    t_i(w^j) \geq t_{i+1}(w^j) \quad \text{for all $j\in[n]$}\,.
\end{equation}
Equivalently, one is asking the type of every prefix string to be a partition. Lets check if $x$ is a YW:
\[
\begin{array}{rcl}
x^1=x_1=1\,, &\quad & t(x^1)=(1) \vdash 1\\[1mm]
x^2=x_1x_2=1\,1 \,,&\quad & t(x^2)=(2) \vdash 2\\[1mm]
x^3=x_1x_2x_3=1\,1\,2 \,,&\quad & t(x^3)=(2,1) \vdash 3\\[1mm]
x^4=x_1x_2x_3x_4=1\,1\,2\,1 \,,&\quad & t(x^1)=(3,1) \vdash 4\\[1mm]
x^5=x_1x_2x_3x_4x_5=1\,1\,2\,1\,2 \,,&\quad & t(x^1)=(3,2) \vdash 5\\[1mm]
\end{array}
\]
Thus, $x$ is a Yamanouchi word. Let $\YC_n$ be the set of YW of length $n$ with entries in $n$, and let $\YC_n^\lm$ be the subset whose type is $\lm$. YW and SYTs are bijectively related.
\begin{proposition}
For every $\lm \vdash n$, there is a bijection between the sets ${\rm SYT}(\lm)$ and $\YC_n^\lm$.
\end{proposition}
\begin{proof}
One can map a word $x=x_1\cdots x_n$ with $t(x)=\lm$ to a $T\in {\rm SYT}(\lm)$ in the following way: each $x_j$ tell you in which row of the SYT to place the integer $j$. Start with a box with number $1$ which corresponds to $x_1=1$ (only possibility), then place a box with number $2$ in row $x_2$, then place a box with the number $3$ in row $x_3$, and so on. Let us illustrate with our previous example:
\[
\begin{array}{rcl}
x^1=1&\quad\arr\quad & \ntableau{1}  \\[3mm]
x^2=11&\quad\arr\quad & \ntableau{1&2}  \\[3mm]
x^3=112&\quad\arr\quad & \ntableau{1&2\\3}  \\[3mm]
x^4=1121&\quad\arr\quad & \ntableau{1&2 & 4\\3}  \\[3mm]
x=x^5=11212&\quad\arr\quad & \ntableau{1&2 & 4\\3&5}  \\[3mm]
\end{array}
\]

\end{proof}

The relevance of this bijection to us will become clear in next section, where we explain how to build a quantum circuit compiling $S_n$-QFT.

\medskip
\medskip
\medskip

We can obtain a characterization of $S_n$-irreps $r^\lm_{S_n}(g)$ in the following way. First, note that $S_n$ is generated by adjacent transpositions $S_n = \langle \{\sg_i=(i,i+1)\}_{i\in[n-1]} \rangle$. Thus, any $g\in S_n$ can be decomposed into a sequence of adjacent transpositions, i.e., as $g=\prod_{k=1}^N \sg_{i_k}$, and
\begin{equation}
    r^\lm(g) =\prod_{k=1}^N r^\lm(\sg_{i_k}) 
\end{equation}
where $N$ is at most $\binom{n}{2}\in\OC(n^2)$. To see this, think of “sorting” a permutation using adjacent swaps (as in bubble sort). Each adjacent swap removes at most one inversion. The worst-case scenario is the reverse permutation, which has $\binom{n}{2}$ inversions. Instead, the $k$-cycle $(1,2,\cdots,k)$ can be written as $(1,2)\cdots (k-1,k)$, i.e., using $k-1$ adjacent transpositions -- so cycles require $\OC(n)$ adjacent transpositions in worst case.

Once we've decomposed into adjacent transpositions $\sg_i$, we can leverage the fact that their action on irrep basis elements $\ket{T}\in\mc{S}^\lm$ is sparse (see e.g., Eq.~28 in Ref.~\cite{grinko2023gelfand}):
\begin{equation}\label{eq_sn_on_T}
        r^{\lm}_{S_n}(\sg_i) \ket{T} = \frac{1}{\Delta_i(T)}\ket{T} + \sqrt{1-\frac{1}{\Delta_i(T)}^2} \ket{\sg_i \cdot T}
\end{equation}
in terms of the axial distance $\Delta_i(T)= {\rm cont}_{i+1}(T)-{\rm cont}_{i}(T)$, where ${\rm cont}_{i}(T)={\rm col}_i(T)-{\rm row}_i(T)$ is the content of cell $i$ in $T$, the difference between the column and row indexes of the box with $i$ in $T$.

For our purposes, it will be handy to have a re-derivation of Eq.~\eqref{eq_sn_on_T} in terms of Yamanouchi words. The row index of $i$ in $x$ is given by ${\rm row}_i(x)=x_i$, since the letter $x_i$ directly tells you the row number where $i$ should be placed. The column index of $i$ in $x$ is given by ${\rm col}_i(x)=t_i(x^i)-x_i$. To see this, note that the tableau is filled in so that within each row the entries are strictly increasing from left to right. This means that when you place the number $i$ in row $x_i$, you put it in the leftmost free (i.e., next available) box in that row. In other words, the column in which $i$ is placed is determined by how many times the row $x_i$ has already been used for earlier entries. Finally, the content of $i$ in $x$ is given by
\begin{align}
    {\rm cont}_i(x) &= {\rm col}_i(x) - {\rm row}_i(x) \\
    &= t_i(x^i) - 2x_i\,,
\end{align}
and the axial distance is
\begin{align}
    \Delta_i(x) &= {\rm cont}_{i+1}(x) - {\rm cont}_i(x)  \\
    &= (t_{i+1}(x^{i+1})-t_i(x^i)) + 2(x_i-x_{i+1})\,.
\end{align}
We thus have
\begin{equation}\label{eq_sn_on_x}
        r^{\lm}_{S_n}(\sg_i) \ket{x} = \frac{1}{\Delta_i(x)}\ket{x} + \sqrt{1-\frac{1}{\Delta_i(x)}^2} \ket{\sg_i \cdot x}
\end{equation}
where $\ket{\sg_i\cdot x} = \ket{x_1 \cdots x_{i+1}x_i \cdots x_n}$. Note that $\sg_i\cdot x$ could, or could not, be a Yamanouchi word. In case its not, the coefficient $\sqrt{1-\frac{1}{\Delta_i^2}}$ will vanish. We will typically encounter the situation where, given a Yamanouchi word $x$ of type $\lm\vdash n$, we are asked to produce the unique Yamanouchi word $x_\mu$ of type $\mu \in \lm+ \Box$ with prefix $x_\mu^n = x$. This is illustrated by the following example:
\begin{example}
  Let $\lm = (3,2)$ and $\mu = (3,3) \in \lm +\Box$. Then: \begin{align} 
  x = 12112  &\mapsto x' = 121122, &
  x = 12121  &\mapsto x' = 121212  \,.
  \end{align}
\end{example}

\subsection{The hook-length formula}
The dimension $d_\lambda$ of an $S_n$ irreducible representation $\lambda$ can be computed in time $\OC(n^2)$ using the hook-length formula. It states that: 
\begin{align}
    d_\lambda &= \frac{n!}{\prod_{(i,j) \in \lambda} h_\lambda(i,j)},
    \label{Eq:hooklength}
\end{align}
where the product runs over all boxes in $\lambda$ and $h_\lambda(i,j)$ is the hook-length of the $(i,j)$-th box: the number of boxes to the right of the box at position $(i,j)$ in $\lambda$ plus the number of boxes below $(i,j)$ plus $1$. Although we state the runtime of our algorithm in terms of $S_n$ irrep dimension ratios -- unfortunately, there is not closed analytic formula for the runtime as a function of $n$ --, we remark that given some input partitions one should be able to use Eq.~\eqref{Eq:hooklength} to efficiently estimate the number of samples needed and therefore the total runtime of the algorithm.

\subsection{Quantum Fourier Transforms over $S_n$, $S_\mu$ and $S_a \wr S_b$}

\paragraph{$S_n$-QFT.--} See Refs.~\cite{Diaconis90, Clausen89} for the classical algorithm and Ref.~\cite{beals1997quantum} for the quantum algorithm. A $S_n$-QFT algorithm is also derived in Ref.~\cite{kawano2016quantum}. Beals' approach is also described in the appendix of Ref.~\cite{bravyi2023quantum} and the following description should complement that exposition. The algorithm crucially uses Young's rule and subgroup-adaptation of the Young's Orthogonal Basis:
\begin{fact}[Young's Rule and Young's Orthogonal Basis~\cite{James84}]
Let $r^\lambda_{S_n}$ be an irreducible representation of $S_n$. Its restriction to $S_{n-1}$ contains $S_{n-1}$-irreducible representations with multiplicity at most $1$ in the following way: 
\begin{align}
    U \left( r^{\lambda}_{S_n}(h) \downarrow^{S_n}_{S_{n-1}} \right) U^\dag &= \bigoplus_{\nu \in \lambda - \Box} r^{\nu}_{S_{n-1}}(h), \, \forall h \, \in S_{n-1}.
    \label{Eq:branching}
\end{align}
Here $\nu - \Box$ is the set of irreducible representations of $S_{n-1}$ that can be obtained by deleting one box from the $\nu$ (seen as a Young diagram) in all permissible ways. Moreover, there exists a basis for $r^\lambda_{S_n}$ and $r^\nu_{S_{n-1}}$ such that $U = I$ in Eq.~\eqref{Eq:branching}: such basis is called the Young's Orthogonal Basis and the property is called subgroup-adaptation. 
\label{lemma:youngs_rule}
\end{fact}

\begin{example} An irreducible representation $r^{(3,2,1)}_{S_6}$ restricted to $S_5$ branches into $\nu \in \lbrace (2,2,1), (3,1,1), (3,2) \rbrace$ irreps of $S_5$.  This means that in Young's orthogonal basis: \begin{align} r^{(3,2,1)}_{S_6} \downarrow^{S_6}_{S_5} = r^{(2,2,1)}_{S_5} \oplus r^{(3,1,1)}_{S_5} \oplus r^{(3,2)}_{S_5}. \end{align}
\end{example}
Fixing a coset representative
$t \in S_n/S_{n-1}$, there is a unique $h \in S_{n-1}$ such that $g = t h$. Fact~\ref{lemma:youngs_rule} then implies that: 
\begin{align}
r^{\lambda}_{S_n}(t h) = r^{\lambda}_{S_n}(t) r^{\lambda}_{S_{n}}(h) = r^{\lambda}_{S_n}(t) \left[r^{\lambda}_{S_{n}}(h) \downarrow^{S_n}_{S_{n-1}}\right] = r_{S_n}^{\lambda}(t) \left[ \bigoplus_{\nu \in \lambda - \Box} r^{\nu}_{S_{n-1}}(h) \right],
\label{App:FourierFirstStep}
\end{align}
One can recurse this as: 
\begin{align}
    r^{\lambda_n}_{S_n}(g) &= r^{\lambda_n}_{S_n}(t_n) \left[ \bigoplus_{\lambda_{n-1} \in \lambda_{n}-\Box} r^{\lambda_{n-1}}_{S_{n-1}}(t_{n-1}) \right] \left[\bigoplus_{\lambda_{n-2} \in \lambda_{n-1} - \Box} r^{\lambda_{n-2}}_{S_{n-2}}(t_{n-2}) \right] \ldots \left[ \bigoplus_{\lambda_{1} \in \lambda_{2}-\Box} r^{\lambda_1}_{S_1}(t_1) \right], 
    \label{App:Fourier}
\end{align}
for $g = t_n t_{n-1} \ldots t_1$ and $r^{\lambda_1}_{S_1}(t_1) = 1$. This recursion is at the heart of the algorithm which we now describe in detail.

\paragraph{$S_n$-QFT algorithm.} Let $f:S_n\arr \mbb{C}$ and define, for any $j\in[n]$, $F^j:S_n\arr \mbb{C}$ such that
\begin{equation}\label{eq_Fj}
    F^j(g)=
    \begin{cases}
        f(g) \text{ if $g\in t_n^j S_{n-1}$}\\
        0 \text{ else,}
    \end{cases}
\end{equation}
where $t_n^j$ is a transversal for the $j$-th coset of $S_{n-1}$ on $S_n$. We denote such space of cosets $T_n\equiv S_n/S_{n-1}$.
Using the set $\{ F_j \}_j$ in Eq.~\eqref{eq_Fj} we can write $\ket{f}=\sum_{g\in S_n} f(g) \ket{g} = \sum_{j=1}^n \ket{F^j}$. Let us also define, for each $j\in[n]$, $f^j_{n-1}:S_{n-1}\arr\mbb{C}$ such that $f^j_{n-1}(h)=f(t_n^jh)$ for all $h \in S_{n-1}$. 

Consider a \textit{transversals register} $\HC_{T_n} = \spn\{ \ket{t_n^1},\cdots,\ket{t_n^n},\ket{*}\}$, where we use the cycles $t_n^i\equiv(i,i+1,\cdots,n-1,n)$ as a choice of transversals for the cosets of $S_{n-1}$ in $S_n$, and where $\ket{*}$ is a dummy state. Once a choice of transversals is made, any $g\in S_n$ has a unique decomposition $g=t_n t_{n-1}\cdots t_2$, where $t_n\in T_n$, $t_{n-1}\in T_{n-1}$ and so on. Note the transversals register $\HC_{T_n}$ is $n+1$-dimensional, and denote $\HC_{T_n}^{S_n/S_{n-1}}$ the $n$-dimensional subspace that excludes the span of the dummy state: $\HC_{T_n}^{S_n/S_{n-1}} = \spn\{ t_n^j\}_{j\in[n]} \cong \mbb{C}[ S_n/S_{n-1} ]$. Consider the composition of multiple transversal registers for the tower $S_n\supset S_{n-1} \supset \cdots \supset S_2$ into
\begin{equation}
\HC_{T^n} = \HC_{T_n} \otimes \HC_{T_{n-1}}\otimes \cdots \otimes \HC_{T_2}
\end{equation}
and note that it contains a subspace $\HC_{T^n}^{S_n} = \HC_{T^n}^{S_n/S_{n-1}} \otimes \cdots \otimes \HC_{T^2}^{S_2/S_1} \cong  \mbb{C}[S_n] \subset \HC_{T^n}$. That is, any $f:S_n\arr\mbb{C}$ can be encoded as a vector $\ket{f}_{T^n} = \sum_g f(g) \ket{g}_{T^n} \in \HC_{T^n}^{S_n}$, with each $\ket{g}_{T^n} = \ket{t_n}_{T_n} \otimes \cdots \otimes \ket{t_2}_{T_2}$ corresponding to its unique factorization $g=t_n\cdots t_2$.

Next, consider the following \textit{word registers} $\HC_{W_j} = \HC_{X_j} \otimes \HC_{Y_j} = \spn\{ \ket{x_j} \otimes \ket{y_j} \}_{x_j,y_j = 0}^j\,$
where again $x_0=y_0=*$ are meant to be dummy registers, such that the composition of word registers 
\begin{equation}
\HC_{W^n} = \HC_{W_n}  \otimes \HC_{W_{n-1}}\otimes \cdots \otimes \HC_{W_2}
\end{equation}
also contains a subspace $\HC_{W^n}^{S_n}$ isomorphic to $\mbb{C}[S_n]$, given by the span of all valid Yamanouchi words (see Eq.~\eqref{eq_Yamanouchi}):
\begin{equation}
\HC_{W^n} \supset \HC_{W^n}^{S_n} =\spn\{ \ket{x,y}\}_{x,y\in \YC_n}  \cong \mbb{C}[S_n]\,.
\end{equation}
In addition, define the subspace $\HC_{W^n}^{S_{n-1}} = \spn \{ \ket{x^{n-1},y^{n-1}}_{W^{n-1}}\otimes \ket{*,*}\}_{x,y\in \YC_n} \cong \mbb{C}[S_{n-1}] \subset \HC_{W^n}$ and note that $\HC_{W^n}^{S_{n-1}}$ and $\HC_{W^n}^{S_{n}}$ are orthogonal. Finally, for each $t_n^i\in T_n$, let
\begin{equation}\label{eq_Hn}
\HC_i = \spn \Big\{ \ket{ \widehat{t_n^i h}}_{W^n} \, ,\quad h\in T^{n-1}\cong S_{n-1} \Big\} \cong \mbb{C}[S_{n-1}] \subset \HC_{W^n}^{S_{n}} \,,
\end{equation}
where $\HC_i \cong \mbb{C}[S_{n-1}]$ and  
\begin{equation}
\ket{\,\widehat{g}\,}_{W^n} = V_{{\rm QFT}_n} \ket{g}_{T^n} =  \sum_{\lm\vdash n} \sqrt{\frac{d_\lm}{n!}}\sum_{x,y \in \YC^\lm_n} [r^\lm_{S_n}(g)]_{x,y}  \ket{x,y}_{W^n}\,,
\end{equation}
Here, $V_{{\rm QFT}_n}: \HC_{T^n}^{S_n} \arr \HC_{W^n}^{S_n}$ is an isometry implementing the desired Fourier transform over $S_n$. Note that these $\HC_i$ subspaces are isomorphic to $\HC_{W^n}^{S_{n-1}}$, and orthogonal to it.

We prove by induction the $S_n$-QFT construction works. 
Starting with an arbitrary $\ket{f} =  \sum_{i=1}^n \ket{t_n^i}_{T_n} \otimes \ket{f^{n-1}_i}_{T^{n-1}}  \in \HC_{T^n}^{S_n}$, we assume access to $V_{{\rm QFT}_{n-1}}$ and use it on $\ket{f}$ to prepare $\sum_{i=1}^n \ket{t_n^i}_{T_n} \otimes \ket{ \hat{f}^{n-1}_i}_{W^{n-1}}$. Following, we tensor with an ancillary word register $\HC_{W_n}$ initialized in its dummy state $\ket{**}_{W_n}$ to get
\begin{equation}
\ket{\phi_0} = \sum_{i=1}^n \ket{t_n^i}_{T_n} \otimes \ket{ \hat{f}^{n-1}_i,**}_{W^{n}}\,.
\end{equation}

For each $l \in [n]$, consider the unitary operation $Z_l=R_l V_E S_l V_E R_l \in \U(\HC_{T_n}\otimes\HC_{W^n})$, where $R_l$, $V_E$ and $S_l$ are specified in Definition~\ref{box_ops}. The following lemma considers applying the sequence $Z_n\cdots Z_1$ on $\ket{\phi_0}$. 
\begin{lemma}[From Ref.~\cite{bravyi2023quantum}]\label{lemma1}
The state after $Z_l\cdots Z_1$,  
$\ket{\phi_l} \equiv Z_l \cdots Z_1 \ket{\phi_0}$, is
\begin{equation}
    \ket{\phi_l} = \sum_{j\in [1,l-1]} \ket{*}_{T_n} \otimes \ket{\hat{F}_j}_{W^n} + \ket{t_n^l}_{T_n}\otimes \ket{ \hat{f}^{n-1}_l, **}_{W^n} + \sum_{j\in[l+1,n]} \ket{t_n^j}_{T_n}\otimes \ket{ \hat{f}^{n-1}_j, **}_{W^n}\,.
\end{equation}
\end{lemma}

Having established that $Z_l \cdots Z_1 \ket{\phi_0}$ produces $\ket{\phi_l}$, we conclude that $Z_n \cdots Z_1$ produces
\begin{equation}
\ket{\phi_{n}} = \sum_{i=1}^n \ket{*}_{T_n} \otimes \ket{\hat{F}_i}_{W^n} = \ket{*}_{T_n}\otimes \ket{\hat{f}}_{W^n}\,.
\end{equation}
We return the register $W^n$ which contains the desired $\ket{\hat{f}}$. A sketch of the circuit implementation is provided in Fig.~\ref{fig:sketch}.

\paragraph{Details and Proof of Lemma~\ref{lemma1}.} We now give details on the iteration and then follow with the proof of the lemma.
\begin{definition}\label{box_ops}
Consider the following operations:
\begin{itemize}
\item For any $\pi \in S_n$, let $R(\pi)$ be the unitary in $\U(\HC_{W^n})$ acting nontrivially only in the subspace $\HC^{S_n}_{W^n}$, where its action on a basis $\{\ket{x,y}\}_{x,y\in \YC_n}$ of $\HC^{S_n}_{W^n}$ is
\begin{align}
    R(\pi)\ket{x,y}_{W^n} &= r^\lm(\pi)\ket{x,y}_{W^n} = \ket{\pi\cdot x,y}_{W^n}\,.
\end{align}
where $\lm$ is the type of $x$ (and of $y$). The action $\pi \cdot x$ of a permutation $\pi$ on a word $x$ is described in the previous section. We will use the shorthand $R_i \equiv R(t_n^i)$ for the transversal $t_n^i$, in which case we have that $R_i$ is an involution -- it is own inverse, $R_i\ad=R_i$ and thus $R_i^2 = I_{W^n}$. Let us stress that $R(\pi)$ acts trivially outside of $\HC_{W^n}^{S_n}$, in particular on $\HC_{W^n}^{S_{n-1}}$:
\begin{equation}
    R(\pi)  \Big( \ket{\hat{f}^{n-1}_i}_{W^{n-1}}\otimes \ket{**}_{W_n} \Big) = \ket{\hat{f}^{n-1}_i}_{W^{n-1}}\otimes \ket{**}_{W_n}\,.
\end{equation}
That is, one can use a control on the $W_n$ register not being $\ket{**}$. The natural way to implement $R_i$ is by decomposing $t_n^i$ into adjacent transpositions as 
\begin{equation}
t_n^i =(i,i+1,\cdots,n-1,n) = (i,i+1)(i+1,i+2)\cdots (n-1,n)\,.
\end{equation}
Thus, we write $t_n^i$ in terms of $n-i-1$ adjacent transpositions each of which have $2$-sparse matrices in Young's Orthogonal representation (see Eq.\eqref{eq_sn_on_T}) and are thus easy to implement. For example,  Ref.~\cite{grinko2023efficient} gives an explicit compilation for the action of adjacent transpositions on Yamanouchi registers using $\OC(n^2)$ elementary gates.

\item Consider the ``embedding'' unitary $V_E$ on $\HC_{T^n}\otimes \HC_{W^n}$ that acts only nontrivially on the subspace $\HC_{T_n}^{S_n/S_{n-1}} \otimes \HC_{W^n}^{S_{n-1}}$, whose action is given by
\begin{align}
     V_E\Big( \ket{t_n^i}_{T_n} \otimes \ket{x^{n-1}y^{n-1},**}_{W^n} \Big) = \ket{t_n^i}_{T_n} \otimes  \sum_{\lm \in \lm_{n-1}+\Box} \sqrt{ \frac{d_\lm}{n d_{\lm_{n-1}}}} \ket{x_\lm,y_\lm}_{W^n} \\
\end{align}
where $x_\lm,y_\lm$ are the unique Yamanouchi words with type $\lm$ and prefixes $x^{n-1}$ and $y^{j-1}$. Note that $\HC_{W^n}^{S_{n-1}}$ and $\HC_{W^n}^{S_{n}}$ are orthogonal but this unitary is effectively rotating the subspace $\HC_{W^n}^{S_{n-1}}$ inside of $\HC_{W^n}^{S_{n}}$, more precisely mapping into $\HC_n \subset \HC_{W^n}^{S_{n}}$ (see Eq.~\eqref{eq_Hn}), provided the register $T_n$ is not in the dummy state. Note that for any $f^{n-1}:\mbb{C}[S_{n-1}] \arr \mbb{C}$, we have
\begin{equation}
V_E\Big( \ket{t_n^i}_{W_n} \otimes \ket{\hat{f}^{n-1}_i,**}_{W^{n-1}} \Big) = \ket{t_n^i}_{W_n} \otimes \ket{\widehat{()f^{n-1}_i}}_{W^n}\,,
\end{equation}
that is, it embeds $\ket{\hat{f}^{n-1}_i,**}_{W^{n-1}} \in \HC_{W^n}^{S_{n-1}}$ into $ \ket{\widehat{()f^{n-1}_i}}_{W^n} \in \HC_n \subset \HC_{W^n}^{S_{n}}$ where $()$ represent the identity permutation. Here, $\widehat{()f^{n-1}_i}$ is the Fourier transform of a function $f$ such $f(g)=f^{n-1}_i(g)$ if $g=()h$ for some $h\in T^{n-1}$ and $f(g)=0$ else. We can write $V_E$ explicitly in therms of the orthonormal bases of $\HC_{W^n}^{S_{n-1}} $ and of $\HC_n$,
\begin{equation}
    V_E = \ketbra{*}{*}_{T_n}\otimes I_{W^n}+ \Big(\sum_{i=1}^n \ketbra{t_n^i}{t_n^i}_{T_n} \Big) \otimes \Big(\sum_{h\in S_{n-1}} \ketbra{\hat{h},**}{\widehat{()h}} +\ketbra{\widehat{()h}}{\hat{h},**} + P_1 \Big)\,.
\end{equation}
where we use $P$ to denote the orthogonal projector $P:\HC_{W^n} \arr \Big(\HC_{W^n}^{S_{n-1}} \oplus \HC_n \Big)^{\perp}$ into the complement (denoted with the $\perp$) of $\Big(\HC_{W^n}^{S_{n-1}} \oplus \HC_n \Big)$ in $\HC_{W^n}$. Note that $V_E$ is too an involution, $V_E=V_E\ad$ and $V_E^2=I_{T_n}\otimes I_{W^n}$.

\item Finally, $S_i$ is a swap between the states $\ket{t_n^i}_{T_n}$ and $\ket{*}_{T_n}$ in the transversals register, conditioned on the state of the word register $\HC_{W^n}$ being within its subspace $\HC_n\subset \HC_{W^n}^{S_n}$. That is, for any $h\in S_{n-1}$,
\begin{align}
    &S_i \Big( \ket{t_n^i}_{T_n}\otimes \ket{\widehat{()h}}_{W^n} \Big) =  \ket{*}_{T_n}\otimes \ket{\widehat{()h}}_{W^n} \\
    &S_i \Big( \ket{*}_{T_n}\otimes\ket{\widehat{()h}}_{W^n} \Big) =  \ket{t_n^i}_{T_n}\otimes \ket{\widehat{()h}}_{W^n}
\end{align}
Again, one can write
\begin{equation}
S_i = \Big(\ketbra{t_n^i}{*} + \ketbra{*}{t_n^i}\Big)\otimes Q + \Big(\sum_{j\neq i} \ketbra{t_n^j}{t_n^j}\Big) \otimes Q^{\perp}
\end{equation}
with $Q$ the orthogonal projector into $\HC_n$ and $Q^\perp$ such that $Q+Q^\perp = I_{W^n}$. Once more, we note $S_i$ is involutory.

\end{itemize}
\end{definition}

We here work out the proof of Lemma~\ref{lemma1}.
\begin{proof}[Proof of Lemma~\ref{lemma1}]
We use induction: explicitly apply $Z_l$ on $\ket{\phi_{l-1}}$ and show one gets $\ket{\phi_l}$. Let us divide the terms in $\ket{\phi_{l-1}}$ into three types:
\begin{align}
&\encircle{A}: \ket{*}_{T_n} \otimes \ket{\hat{F}_j} \text{ for $j\in[1,l-1]$}\\
&\encircle{B}: \ket{t_n^l}_{T_n} \otimes \ket{\hat{f}_l^{n-1},**}\\
&\encircle{C}: \ket{t_n^j}_{T_n} \otimes \ket{\hat{f}_j^{n-1},**}\text{ for $j\in[l+1,n]$}
\end{align}

Note that in order to get $\ket{\phi_l}$, $Z_l$ needs to act trivially on $\encircle{A}$ and $\encircle{C}$ terms in $\ket{\phi_{l-1}}$, while taking $\ket{t_n^l}_{T_n} \arr\ket{*}_{T_n} \otimes \ket{\hat{F}_l}$.

Lets see what it does on each type of term (see Fig.~\ref{fig:lemma1} for a sketch). On a type $\encircle{A}$-term:
\begin{align}
R_lV_ES_lV_ER_l \Big( \ket{*}_{T_n} \otimes \ket{\hat{F}_j} \Big)&= R_lV_ES_lV_E \Big( \ket{*}_{T_n} \otimes \ket{\hat{F}_{lj}}_{W^n} \Big) \\
&=R_lV_ES_l  \Big( \ket{*}_{T_n} \otimes \ket{\hat{F}_{lj}}_{W^n} \Big) \\
&=R_lV_E \Big( \ket{*}_{T_n} \otimes \ket{\hat{F}_{lj}}_{W^n} \Big)\\
&=R_l  \Big( \ket{*}_{T_n} \otimes \ket{\hat{F}_{lj}}_{W^n} \Big) \\
&= \ket{*}_{T_n} \otimes \ket{\hat{F}_j}_{W^n}\,,
\end{align}
where $R_l$ took $\ket{\hat{F}_j} \in \HC_j \subset \HC_{W^n}^{S_n}$ into $\ket{\hat{F}_{lj}} \in \HC_{lj} = \spn\{ \ket{\widehat{t_n^lt_n^j h}}\}_{h\in S_{n-1}}$, then $V_E$ acted trivially (because of the control on $\ket{*}_{T_n}$), then $S_l$ also acted trivially (because $\ket{\hat{F}_{lj}} \not\in \HC_n$), then again $V_E$ acts trivially and $R_l$ takes $\ket{\hat{F}_{lj}}$ back to $\ket{\hat{F}_{j}}$.

On the type $\encircle{B}$-term:
\begin{align}
R_lV_ES_lV_ER_l \Big( \ket{t_n^l}_{T_n} \otimes \ket{\hat{f}^{n-1}_l, **}_{W^n}  \Big)&= R_lV_ES_lV_E \Big( \ket{t_n^l}_{T_n} \otimes \ket{\hat{f}^{n-1}_l, **}_{W^n} \Big) \\
&=R_lV_ES_l \Big(\ket{t_n^l}_{T_n} \otimes\ket{ \widehat{() f^{n-1}_l} }_{W^n}\Big) \\
&=R_lV_E \Big(\ket{*}_{T_n} \otimes\ket{ \widehat{() f^{n-1}_l} }_{W^n} \Big)\\
&=R_l \Big(\ket{*}_{T_n} \otimes\ket{ \widehat{() f^{n-1}_l} }_{W^n} \Big) \\
&=\ket{*}_{T_n} \otimes\ket{ \hat{F}_l }_{W^n} \\
\end{align}
where, in order, we first used that $R_l$ acts trivially outside $\HC_{W^n}^{S_n}$, then $V_E$ embedded $\ket{\hat{f}^{n-1}_l, **}$ into $\ket{ \widehat{() f^{n-1}_l} } \in \HC_n \subset \HC_{W^n}^{S_n}$, after which $S_l$ acted non-trivially (both conditions are met) mapping $\ket{t_n^l}_{T_n} \arr \ket{*}_{T_n}$, then $V_E$ acted trivially (because of the control on $\ket{*}_{T_n}$), and finally $R_l$ pushed $\ket{ \widehat{() f^{n-1}_l} }$ into $\ket{\hat{F}_l }$.

On a type $\encircle{C}$-term:
\begin{align}
R_lV_ES_lV_ER_l \Big( \ket{t_n^j}_{T_n} \otimes \ket{\hat{f}^{n-1}_j, **}_{W^n}  \Big)&= R_lV_ES_lV_E \Big( \ket{t_n^j}_{T_n} \otimes \ket{\hat{f}^{n-1}_j, **}_{W^n} \Big) \\
&=R_lV_ES_l \Big(\ket{t_n^j}_{T_n} \otimes\ket{ \widehat{() f^{n-1}_j} }_{W^n}\Big) \\
&=R_lV_E \Big(\ket{t_n^j}_{T_n} \otimes\ket{ \widehat{() f^{n-1}_j} }_{W^n}\Big)\\
&=R_l \Big( \ket{t_n^j}_{T_n} \otimes \ket{\hat{f}^{n-1}_j, **}_{W^n} \Big) \\
&=\ket{t_n^j}_{T_n} \otimes \ket{\hat{f}^{n-1}_j, **}_{W^n} \,, \\
\end{align}
where, for the same reasons as above $R_l$ acts trivially outside $\HC_{W^n}^{S_n}$ and $V_E$ embeds $\ket{\hat{f}^{n-1}_j, **}$ into $\ket{ \widehat{() f^{n-1}_j} } \in \HC_n \subset \HC_{W^n}^{S_n}$. Now, instead, $S_l$ acts trivially (since $t_n^j\not\in\{*,t_n^l\}$) and therefore $V_E$ un-embeds the $W^n$-register. Finally, $R_l$ acts trivially since the $W_n$ register is dummy on $\ket{\hat{f}^{n-1}_j,**}$.

\end{proof}

\begin{figure}[h]
    \centering
    \includegraphics[width=\textwidth]{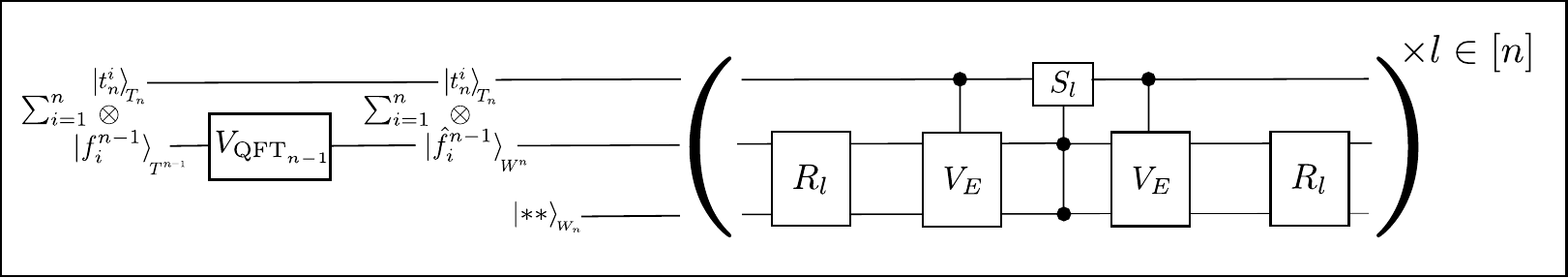}
    \caption{Schematic circuit for $S_n$-QFT.}
    \label{fig:sketch}
\end{figure}

\begin{figure}[h]
    \centering
    \includegraphics[width=\textwidth]{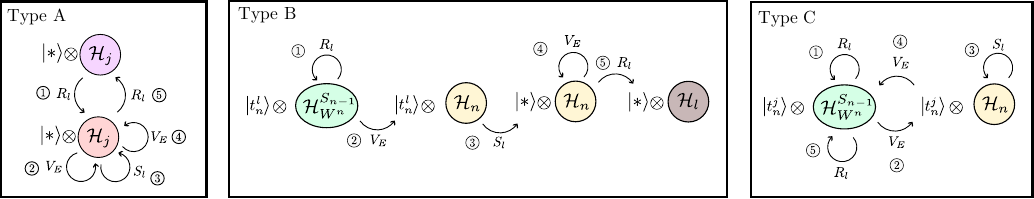}
    \caption{Sketch for the action of $Z_l= R_l V_ES_l V_E R_l $ on type $A$, $B$ and $C$ terms in $\ket{\phi_{l-1}}$.}
    \label{fig:lemma1}
\end{figure}

\paragraph{$S_\mu$-QFT and $S_a \wr S_b$-QFT.} Recall that the irreducible representations of direct product groups are tensor factors of irreducible representations of the product factors. A Fourier transform over $S_\mu$-QFT, can be implemented as a QFT over each $S_{\mu_i}$ in $S_\mu$ separately and taking a tensor product. The QFT for $G = S_a \wr S_b$ was derived in Moore, Rockmore and Russel \cite{Moore07}.

\subsection{Generalized Phase Estimation over Restricted Representations}
Our algorithm is a variation on Harrow's generalized phase estimation (GPE) when applied to restricted representations~\cite{harrow2005applications}. 
In our algorithm (see Fig.~\ref{fig_1}), we choose $F=H\subset G$, fixing $\CC = \mbb{C}[H]$, and select a target register $\TC = \mbb{C}[G]$. That is, $\TC$ holds a left-regular $G$-rep but, in this context, its best to regard $\TC$ as an $H\subset G$-rep (by restriction). In this target register $\TC$, we initialize: 
\begin{align}
    \rho^\a_{G} &= 
U^\dag \left( \ket{\a}\bra{\a} \otimes \frac{I_{d_\a}}{d_\a} \otimes \frac{I_{d_\a}}{d_\a}\right) U = \frac{\Pi_\a^G}{d_\a^2},
\end{align}
where $U$ is the $G$-QFT and initialize $\CC$ to a uniform superposition and apply a controlled gate $\sum_{h\in H} [\ketbra{h}{h}]_\CC\otimes [R(h)]_\TC$, where as clarified above $R\equiv R_\TC$ is the $H$-rep arising from the restriction of the regular representation of $G$. We then weak-Fourier-sample on the control register, 
finding the outcome $\b \in \widehat{H}$ with probability $p_\a(\b) = \text{Tr}\left(E_\b \rho^\a_G E_\b^\dag \right)$, 
where 
\begin{equation}
    E_\b = \frac{1}{\sqrt{|H|}}\sum_{h \in H} \Bigg[ \Big(\ketbra{\b}{\b} \otimes I_{d_\b}\otimes I_{d_\b} \Big) V \ket{h} \Bigg]_{\mc{C}} \otimes \Big[ R(h)\Big ]_{\mc{T}}
\end{equation}
with $V$ the $H$-QFT and $\ketbra{\b}{\b} \otimes I_{d_\b}\otimes I_{d_\b}$ the measurement of irrep label $\b$ (and discard of the other two registers of multiplicity and dimension, after $H$-QFT).
Note that $[r^{\b}_H(h)]_{ij} = \braket{i}{r^{\b}_H(h)\Big|j}$ and that since $r_G^\a(h)$ is an irrep of $G$ and only enters evaluated on elements of a subgroup $H$, it corresponds to the restriction of the $G$-irrep to the subgroup $H \subseteq G$. The outcome probability becomes:
\begin{equation}
p_\a(\b) = \Tr[E_\b \rho^\a_G E_\b\ad ] = \frac{1}{d_\a^2} \sum_{h,g \in H} \Tr[E_\b \Pi_\a^G E_\b\ad  ]  = \frac{1}{d_\a^2 |H|} \sum_{h,g \in H}  \Tr[ \Pi_\b^H \ketbra{h}{g}] \Tr[ R(h) \Pi^G_\a R(g)\ad]\,.
\end{equation}
where $\Pi_\b^H = V\ad(\ketbra{\b}{\b} \otimes I_{d_\b}\otimes I_{d_\b})V$. Consider the first term
\begin{align}
\Tr[ \Pi_\b^H \ketbra{h}{g}] &= \Tr[ \Big(\ketbra{\b}{\b}\otimes I_{d_\b}\otimes I_{d_\b}\Big) \Big( V \ketbra{h}{g} V\ad\Big)] \\
&= \sum_{\g,\w} \sum_{i,j,k,l\in [d_\b]} \frac{d_\b}{|H|}  [r^\g_H(h)]_{i,j} [r^\w_H(g)]^*_{kl}  \Tr[ \Big(\ketbra{\b}{\b}\otimes I_{d_\b}\otimes I_{d_\b}\Big) \Big(  |\g,i,j \rangle\langle \w,k,l |  \Big) ] \\
& =\frac{d_\b}{|H|}\sum_{i,j \in [d_\b]} [r_H^\b(h)]_{ij} [r_H^\b(g)]^*_{ij}\\
& =\frac{d_\b}{|H|}\sum_{i,j \in [d_\b]} [r_H^\b(h)]_{ij} [r_H^\b(g^{-1})]_{ji}\\
& =\frac{d_\b}{|H|}\sum_{i \in [d_\b]}  [r_H^\b(hg^{-1})]_{ii}\\
&=\frac{d_\b}{|H|} \chi^\b_H(hg^{-1})\,.
\end{align}
Similarly $\Tr[ R(h) \Pi^G_\a R(g)\ad] = \Tr[ \Pi^G_\a R(g^{-1}h)] = d_\a \chi^\a_G(g^{-1}h)$, so overall we have
\begin{align}
p_\a(\b) &= \Tr[E_\b \rho^\a_G E_\b\ad ] \\
&= \frac{1}{d_\a^2|H|} \sum_{h,g \in H} \frac{d_\b d_\a}{|H|} \chi^\b_H(hg^{-1}) \chi^\a_G(g^{-1}h) \\
&= \frac{1}{d_\a^2|H|} \sum_{h,g \in H} \frac{d_\b d_\a}{|H|} \chi^\b_H(hg^{-1}) \chi^\a_G(hg^{-1})^* \\
&=\frac{d_\b}{d_\a} \mult{r^\b_H}{r^\a_G \downarrow^{G}_H}\,.
\end{align}
The last equality follows from expanding the character of the reducible $H$-rep $R\coloneqq r^\a_G\downarrow^G_H$ into irreducible characters
$\chi^R_H = \sum_{\g} \mult{r^\g_H}{r^R_H} \,\chi^\g_H$ and then using the orthonormality of irreducible characters $\langle \chi^\g_H,\chi^\w_H \rangle= \frac{1}{|H|}\sum_{h\in H} \chi^\g_H(h) \chi^\w(h)^* =\d_{\g\w}$. We now closely analyze the number of samples we need to take from GPE to find the multiplicity with a constant probability.

 \begin{lemma}
 \label{Lemma:shots}
Algorithm \ref{theorem_quantum_algo} finds $\mult{r^\b_H}{r^\a_G \downarrow^{G}_H}$ with constant probability after $\OC((d_\alpha/d_\beta)^2)$ shots.
 \end{lemma}
 \begin{proof}
     In each run, the algorithm samples the output distribution: 
     \begin{align}
         p(\beta') &= \frac{d_\beta'}{d_\alpha} \mult{r^{\b'}_H}{r^\a_G \downarrow^{G}_H}
     \end{align}
     Let $x_i$ be a boolean random variable such that $x_i = 1$ if the $i$-th sample is $\beta$ and $x_i = 0$ otherwise. It follows that $\mathbb{E}[x_i] = \frac{d_\beta}{d_\alpha} \mult{r^{\b}_H}{r^\a_G \downarrow^{G}_H}$. Let $N$ be the number of samples taken from the distribution and let $X = \sum_i x_i$, so that $\mathbb{E}[X] = N \mathbb{E}[x_i]$. Chernoff-Hoeffding inequality gives:
    \begin{align}
        \text{Pr} \left[ \left| X - \mathbb{E}[X] \right| > a \right] = \text{Pr} \left[ \left| \frac{d_\alpha}{N d_\beta} X - \mult{r^{\b}_H}{r^\a_G \downarrow^{G}_H} \right| > \frac{d_\alpha}{N d_\beta} a \right] \leq 2 \exp\left( -\frac{2a^2}{N} \right).
    \end{align}
    Because $\mult{r^{\b}_H}{r^\a_G \downarrow^{G}_H}$ is an integer, then as long as $ a d_\alpha / d_\beta N <  1/2$, we can round $X d_\alpha/ d_\beta N$  to the nearest integer and get the correct value of the multiplicity coefficient. We can thus set $a = d_\beta N/2d_\alpha$ and set: 
    \begin{align}
         2 \exp\left( -\frac{2a^2}{N} \right) = C
    \end{align}
    for some small number $C$ (say $C = 1/100$ if we want to have $99\%$ chance of getting the correct answer). But this gives: 
    \begin{align}
        N \leq (d_\alpha/d_\beta)^2 \ln(4/C^2).
    \end{align}
    and proves the assertion.
 \end{proof}

\subsection{Kostka Numbers: Details} 
Recall that a partition $\mu\vdash n$ with $l(\mu)=k \leq n$ defines a subgroup
$S_\mu=S_{\mu_1} \times \cdots \times S_{\mu_k} \subseteq S_n$ called a \textit{Young subgroup}.
An irrep of $S_\mu =\bigtimes_i S_{\mu_i}$ is given by the tensor of arbitrary irreps $r^{\alpha_i}_{S_{\mu_i}}$ of each of the factors where $\alpha_i \vdash \mu_i$. The irreducible representation $r^\alpha_{S_\mu}$ is thus labeled by a vector of partitions $\a=(\a_1,\cdots,\a_k)$, where $\a_i\vdash \mu_i$. In particular, the trivial irrep corresponds to $\a=([\mu_1],\cdots,[\mu_k]) := \triv$. 

Let us denote the Specht module (see Eq.~\eqref{eq_weight_mod}) of $r^\triv_{S_\mu}$ by $\mc{S}^\triv_{S_\mu}$.
For each $\mu\vdash n$, the \textit{permutation module} $\mc{M}^\mu$ is the $S_n$-module that is induced from the trivial irrep of the corresponding Young subgroup $S_\mu$, $\mc{M}^{\mu} = (\mc{S}^{{\rm triv}}_{S_\mu} ){\uparrow^{S_n}_{S_\mu}}\,.$
\begin{equation}
    \mc{M}^{\mu} = (\mc{S}^{{\rm triv}}_{S_\mu} ){\uparrow^{S_n}_{S_\mu}}\,.
\end{equation}
This module admits an action with an $S_n$ representation $R^\mu$.
The \textit{Kostka numbers} $K_{\lm}^{\mu}$ are the multiplicity of $S_n$-irreps $\mc{S}^\lm_{S_n}$ in $\mc{M}^{\mu}$:
\begin{equation}
\mc{M}^{\mu} \cong \bigoplus_{\lm} \mbb{C}^{K_{\lm}^{\mu}} \otimes \mc{S}^{\lm}_{S_n}\quad \text{ or }\quad R^\mu\cong \bigoplus_{\lm} I_{K_{\lm}^{\mu}} \otimes r^{\lm}_{S_n} 
\end{equation}
where $R^\mu:S_n\arr \U(\mc{M}^\mu)$ is the representation corresponding to the action of $S_n$ on the permutation module $\mc{M}^\mu$. That is, the Kostka coefficients are the dimension of the space of $S_n$-equivariant homomorphisms between permutation modules and Specht modules, $K_{\lm}^\mu = \text{Hom}_{S_n} (\mc{M}^\mu, \mc{S}^{\lambda}_{S_n})$. We also note that 
\begin{align}
K_\lm^\mu & = \{ \text{$\#$ of irreps of type $\lm$ in permutation module $\mc{M}^\mu$} \} = \frac{\dim( (\mc{M}^\mu)_\lm )}{d_\lm} \\
&= \{ \text{$\#$ of states with weight $\mu$ in weight module $\mc{W}^\lm$} \} = \dim( \mc{W}^\lm_\mu)
\end{align}
where $\HC_\lm$ denotes isotypic $\lm$ in module $\HC$. The Kostka number $K^\mu_\lambda$ is equivalently the number of semistandard Young tableaux of shape $\lambda$ and content $\mu$.

\paragraph{Efficient Classical Algorithm for Kostka Numbers with Bounded Irrep Dimension}
The Kostka number $K^{\mu}_\lambda$ is the number of semi-standard Young tableaux of shape $\lambda$ and content $\mu$. It is known that: $K^\mu_\lambda > 0$ if and only if $\lambda \succeq \mu$ in the dominance order ($\mu \succeq \nu$ if $\mu_1 + \ldots + \mu_r \geq \nu_1 + \ldots + \nu_r$ for all $r$).
 
 \begin{lemma}[Fayers \cite{fayers2019note}] $\mu \succeq \nu \implies K_\lambda^\nu \geq K_\lambda^\mu$. 
 \label{lemma:Fayers}
 \end{lemma} This in particular means that $d_\lambda \geq K_\lambda^\mu$ for all $\mu$, with equality only if $\mu = (1,1,1 \ldots)$ in which case $K_\lambda^{(1,1,1 \ldots)} = d_\lambda$ counts the number of standard tableaux of shape $\lambda$. This is equal to a dimension of the $S_n$ irreducible representation $\lambda$. We use this to give an efficient algorithm that computes the Kostka numbers whenever $d_\lambda \leq poly(n)$.

\begin{algorithm}[H]
\caption{Algorithm that generates all SSYT of shape $\lambda$ and content $\mu$}
\label{alg:SSYT}
\begin{algorithmic}[1]
\Procedure{SSYT}{$\lambda,\mu$} \Comment{generates a set of all SSYTs with shape $\lambda$ and content $\mu$.}
\State $T \gets$ empty Young Tableau of shape $\lambda$
\State $states \gets set( (T, \mu))$ \Comment{\emph{states} is a set of tuples (tableau, content).}
\While{True}
\State $T', \mu' \gets states.pop(0)$\Comment{pop(0) removes the first element of \emph{states} and returns it.}
\State $corners \gets corners(T')$ \Comment{corners($T'$) returns a list of all corners of a tableau $T'$.}
\For{$corner$ in $corners$}
\State $Q \gets T'$
\State $m \gets \text{largest entry in } \mu'$
\State $\text{remove $m$ from } \mu'$ \Comment{for example if $\mu' = (1^3, 2^2, 3)$, then $m=3$ and $\mu' \rightarrow (1^3, 2^2)$}
\State $Q (corner) \gets m$ 
\If{$Q \text{ is SSYT}$} \Comment{checks if $Q$'s nonempty boxes are non-decreasing along rows and increasing along columns.}
\State $states.add((Q, m))$ \Comment{since \emph{states} is a set, all of its elements are unique.}
\EndIf
\EndFor
\If{\text{all \emph{states} have empty content} }
\State \Return  $states$
\EndIf
\EndWhile
\EndProcedure
\end{algorithmic}
\end{algorithm}

 \begin{figure}[h]
        \centering
        \includegraphics[scale=0.35]{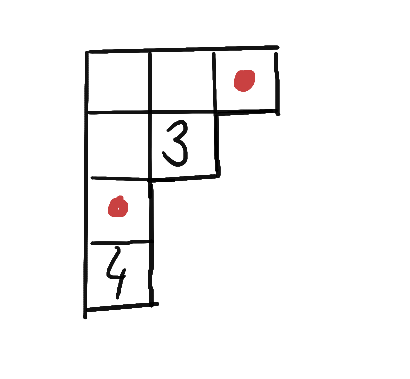}
        \caption{A corner is a box in a partially-filled Young diagram such that there is no box to the right of it or the box to the right of it is non-empty and there is not box below it or the box below it is non-empty. The red dots above label corners in a partially-filled semistandard Young Tableau of shape $\lambda = (3,2,1,1)$.}
        \label{fig:corners}
    \end{figure}

\begin{figure}[h]
        \centering
        \includegraphics[scale=0.35]{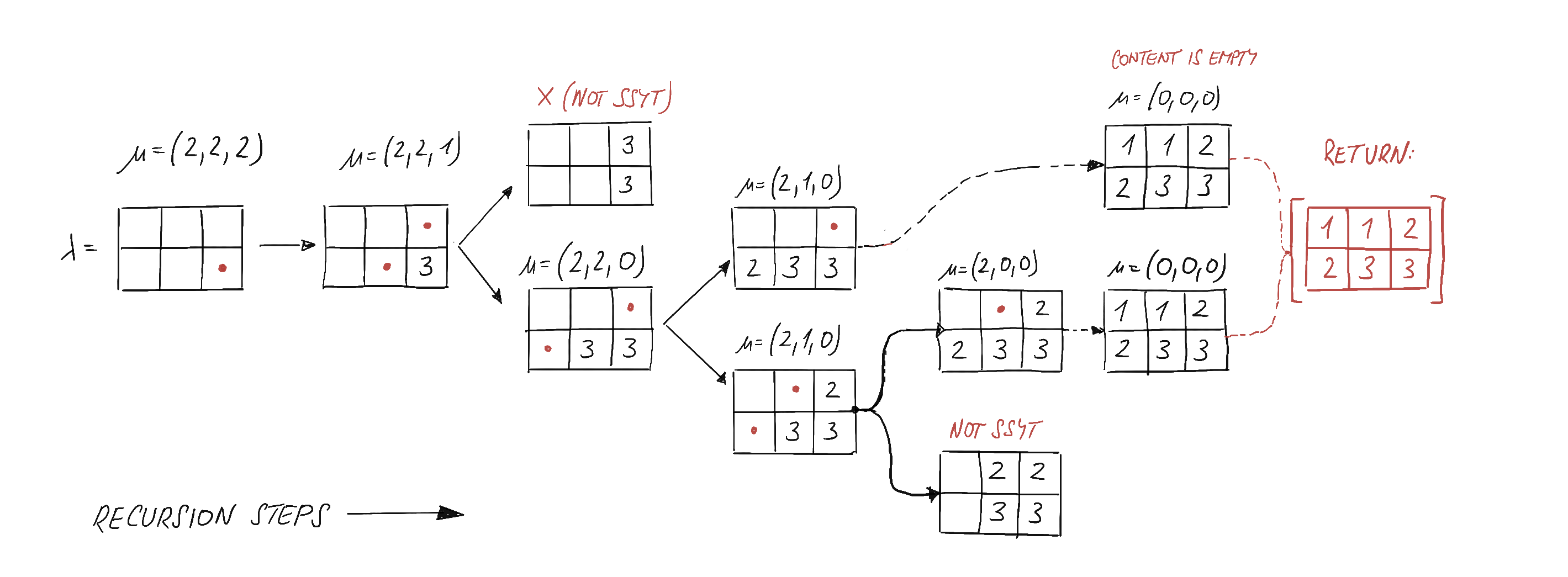}
        \caption{Execution of Algorithm 1 on input $\lambda = (3^2), \mu = (2^3)$.}
        \label{fig:recursion}
    \end{figure}

\begin{figure}[h]
    \centering
    \includegraphics[scale=0.35]{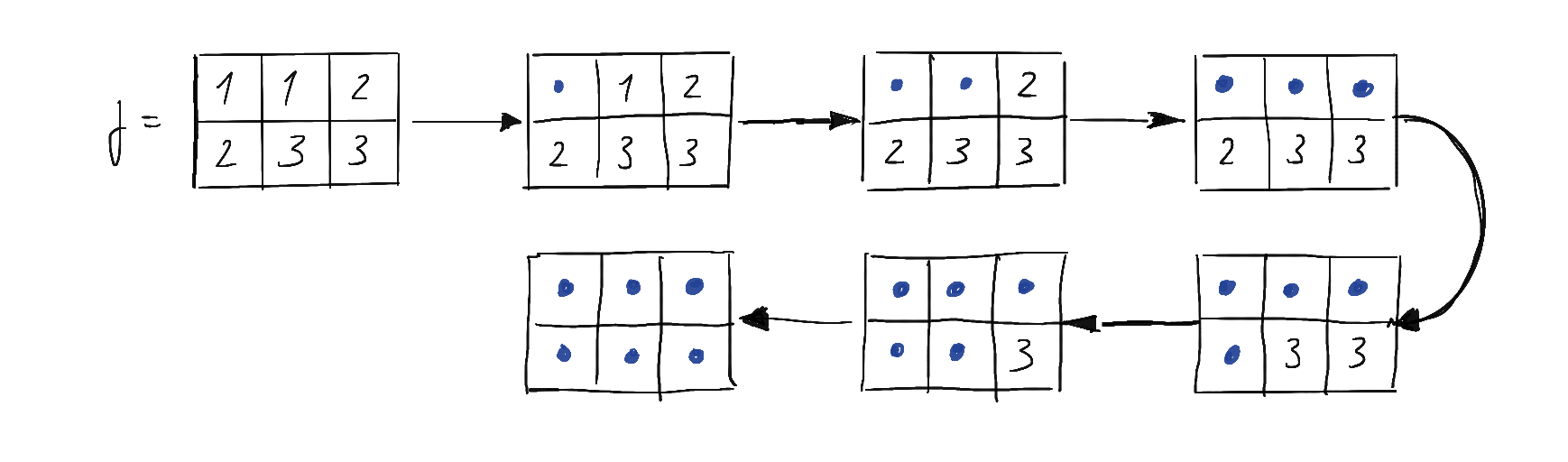}
    \caption{Running the algorithm in reverse on input $[[1,1,2],[2,3,3]]$. The same path, when reversed, appears in Fig~\ref{fig:recursion}}
    \label{fig:backtracking}
\end{figure}

\begin{lemma}
\label{app:kostka_in_FP}
There is a classical polynomial-time algorithm that computes $K^\mu_\lambda$ whenever $d_\lambda \leq \text{poly}(n)$.
\end{lemma}
\begin{proof}
    See Algorithm~\ref{alg:SSYT}. The set \emph{states} stores tuples $(T, \mu)$, where $\lambda$ is a Young tableau that can have some boxes empty.
    A \emph{corner} is a box in the Young diagram such that there is no box to the right of it or the box to the right of it is not empty and there is no box below it or the box below it is non-empty (see Fig.~\ref{fig:corners}). The algorithm always returns semistandard Young tableaux because it checks if the partially filled tableau  satisfies the constraints of SSYT before including it into the set \emph{states}. It also returns all SSYT of shape $\lambda$ and content $\mu$; this can be seen by considering an SSYT $\gamma$ of shape $\lambda$ and content $\mu$ and running the algorithm `in reverse': 
    If the Algorithm~\ref{alg:SSYT} generated $\gamma$, then the set states must have contained a state $(\gamma, ())$. This state was added to $\emph{states}$ if  $\emph{states}$ contained $(T, (1))$, where $T$ is $\gamma$ with its top-left entry removed. By similar logic, it is easy to see that if $(T, \nu)$ was contained in the set \emph{states}, then either it is $(\lambda, \mu)$ or $(T', \nu')$ was contained in states, where $T'$ is $T$ without the smallest element in the top-left and $\nu'$. This shows that the algorithm produces all SSYT of shale $\lambda$ and content $\mu$. It remains to prove the efficiency of the algorithm if $d_\lambda \leq poly(n)$. 
    Notice that Lemma~\ref{lemma:Fayers} implies that there are at most $d_\lambda$ semistandard partitions and it follows that at every recursion step, the set \emph{states} contains at most $d_\lambda$ many Young tableaux. There are at most $n$ recursion steps and it follows that the algorithm runs in polynomial time. See Fig.~\ref{fig:recursion}. 
\end{proof}

While this algorithms suffices for our purposes, \cite{panova2025polynomialtimeclassicalversus} sketched a more efficient algorithm for this task.

\subsection{Littlewood-Richardson~Coefficients: Details}
Let $W$ be a $\SU(d)$-module. For any partition $\lm\vdash n \in \mbb{N}$ and $l(\lm)\leq d$, 
The \textit{Schur functor}\footnote{When we treat the module $V$ as an object in the category of representations of $G$; $\emph{Rep}(G)$ then $\mathbb{S}^\lambda$ is a functor mapping the representation $V$ to the $\lambda$-isotypic component of $V^{\otimes n}$ called the \emph{Schur Functor}.}
$\mathbb{S}^\lambda$ is the map from $W$ to the isotypic component for $S_n$ labeled by $\lm$ in $W\tn$.
That is, $\mbb{S}^\lm(W)$ is a new $\SU(d)$-module that is obtained by first tensoring $n$ copies of $W$ and then projecting into the subspace containing all $S_n$-irreps of type $\lm$
\begin{align}
    \mathbb{S}^\lambda(W) := \Pi_\lambda W^{\otimes n} \subset W\tn.
\end{align}

The choice $W=\mbb{C}^d$, the standard irrep of $\SU(d)$, yields, for the different $\lm$, all irreducible $\SU(d)$-modules $\mbb{S}^\lm(\mbb{C}^d)= \mc{W}^\lm$ (the \textit{weight modules} in Eq.~\eqref{eq_weight_mod}). That is
\begin{equation}
    r^\lm_{\SU(d)} : \SU(d) \arr \U( \mbb{S}^\lm(\mbb{C}^d)) \,.
\end{equation}

The Littlewood-Richardson coefficients appear in the decomposition of the tensor product of $\SU(d)$ irreducible representations
\begin{align}
    r^\lm_{\SU(d)} \otimes r^\mu_{\SU(d)} & \cong \bigoplus_{\nu} I_{c_{\lm\mu}^\nu} \otimes r^\nu_{\SU(d)}.
\end{align}
Crucially (for us), the same coefficients appear in the restriction of the $S_n$-irreps $r_{S_n}^{\nu}$ to the subgroup $S_{a} \times S_{b}$, with $a+b=n$. Specifically, 
\begin{align}
   r_{S_n}^{\nu} \downarrow^{S_n}_{S_a \times S_b} \cong \bigoplus_{\substack{\lm\vdash a \\\mu\vdash b}} I_{c^{\nu}_{\lm,\mu}} \otimes \Big( r^{\lm}_{S_a} \otimes r^{\mu}_{S_b} \Big)
   \,.
\end{align}

\paragraph{Combinatorial Interpretation for LR: Littlewood-Richardson's rule.}
A combinatorial interpretation for LR coefficients was originally proposed by Littlewood and Richardson in 1934, but only proven four decades later~\cite{howe2012should,stembridge2002concise,James84}.
\begin{definition}[LR rule]
The LR coefficient $c_{\lm\mu}^\nu$ equals the number of LR tableau of shape $\nu/\lm$ and content $\mu$.
\end{definition}
We need some background to unpack this statement. A LR tableau is a skew semi-standard tableau with the extra condition that the concatenation of its reversed rows is a Yamanouchi word (every prefix contains at least the same number of $i$'s than of $i+1$'s).
A \textit{skew Young diagram} is formed by subtracting one Young diagram $\lm$ from another `larger' Young diagram $\nu$, provided $\lm$ `fits completely within $\nu$' -- provided $\lm_i\leq \nu_i$ for all $i$. This results in a non-standard `skew' shape $\nu/\lm$ that does not necessarily start at the top-left corner. 
For example, if $\nu=[5,3,1]$ and $\lm=[3,1]$ we have the skew diagram
\begin{equation}\label{eq_skew_diagram}
\nu/\lm=\ydiagram{3+2,1+2,1}\,.
\end{equation}
Here, for example, the first row is formed by skipping the first three boxes (part of $\lm$) and then placing two boxes (remaining part of $\nu$).
A \textit{skew Young tableau} is a filling of a skew Young diagram with integers. As a direct extension of SYTs and SSYTs introduced in Eqs.~\eqref{eq_SYT} and~\eqref{eq_SSYT}, we can have skew SYTs and skew SSYTs. For example, for the skew shape $\nu/\lm=[5,3,1]/[3,1]$, we could have the following standard and semi-standard skew tableaux, respectively
\begin{equation}\label{eq_skew_tableaux}
\ytableausetup{boxsize=normal,aligntableaux=center}
\begin{ytableau}
   \none & \none & \none & 1 & 2 \\
   \none & 3 & 4 \\
   5
\end{ytableau}\,, \quad 
\ytableausetup{boxsize=normal,aligntableaux=center}
\begin{ytableau}
   \none & \none & \none & 1 & 1 \\
   \none & 2 & 2 \\
   3
\end{ytableau}\,.
\end{equation}
The content of such is $\mu=[1^5]$ and $\mu=[2,2,1]$, respectively.
The first tableau is not a LR tableau since its reverse row reading, $21435$, is not a Yamanouchi word -- the prefix $2$ has more $2$'s (one) than $1$'s (zero). Instead the second tableau is a LR tableau since its reverse row reading, $11223$, has, for all prefixes, at least as many $i$ as $i+1$:
- "1" and "11" are trivially Yamanouchi.
- "112" has more 1's than 2's.
- "1122" still has more 1's than 2's.
- "11223" finally, adds a 3, still maintaining more 1's than 2's, and it has just one 3, satisfying the condition. This means that $c_{\lm\mu}^\nu$ is at least 1 for $\nu=[5,3,1]$, $\lm=[3,1]$ and $\mu=[2,2,1]$.

\subsection{Kronecker Coefficients: Details} 
Kronecker coefficients are the multiplicities of $S_n$ irreducible representations in their tensors: 
\begin{align}
    g_{\mu \nu \lambda} &= \mult{r^\lambda_{S_n}}{r^\mu_{S_n} \otimes r^\nu_{S_n}}, \; \mu, \nu, \lambda \vdash n,
\end{align}
that play an important role in representation theory, algebraic combinatorics and many other areas. It is a major open problem (Stanley's Problem 10 \cite{Stanley99}), whether the coefficients can be expressed as a sum of combinatorial objects or not. This question plays an important role in the Geometric Complexity Theory approach to $\P$ vs $\NP$ of Mulmuley and Sohoni \cite{Mulmuley12GCT} and has attracted a lot of interest for this reason (see \cite{pak2015complexity} and references therein). The coefficients also appeared in works on quantum marginal problem (see for example Refs.~\cite{Christandl07, Christandl05, klyachko2004quantummarginalproblemrepresentations}). The coefficients can be also computed using the character formula: 
\begin{align}
    g_{\mu \nu \lambda} &= \frac{1}{n!} \sum_{\pi \in S_n} \chi^\mu(\pi) \chi^\nu(\pi) \chi^\lambda(\pi),
\end{align}
which shows that the coefficients are symmetric under permutation of its inputs. They have been shown to be $\# \P$-hard to compute \cite{Burgisser08}, but are known to be in $\# \P$ on some sets of restricted inputs \cite{pak2015complexity} and in $\P$ for inputs with constant length \cite{Christandl12}. See Ref.~\cite{bravyi2023quantum} for additional background on the coefficients in context of quantum computation. 

\paragraph{Kronecker Coefficients: Inputs with a possible polynomial speedup.}

A nontrivial set of inputs in which the quantum algorithm runs in polynomial time is when $d_\mu \leq \poly(n)$ and $d_\lambda = d_\nu$ is arbitrary. A nice set of partitions that satisfy this condition is $\mu = (n-k, \alpha_1, \alpha_2 \ldots)$, where $\alpha$ is a partition of $k$ (for a constant $k$), such that $\alpha_1 < n-k+1$. A simple application of the hook length formula shows that the dimension of such partition is upper bounded by $(d_\alpha n^k/k!)$, which is polynomial of degree $k$ in $n$.

\begin{lemma}
\label{Lemma:partitions}
    Let $\mu = (n-k, \alpha_1, \alpha_2 \ldots)$, where $\alpha$ is a partition of $k$ (for a constant $k$), such that $\alpha_1 < n-k+1$. Then $d_\mu \leq \frac{k!n^k}{d_\alpha}$.
\end{lemma}
\begin{proof}
Using the hook-length formula, we have:
\begin{align}
    d_\mu &= \frac{n!}{\Pi_{(i,j) \in \mu} h(i,j)} & d_\alpha &= \frac{k!}{\Pi_{(i,j) \in \alpha} h(i,j)},
\end{align}
where $h(i,j)$ is the hook length between boxes $i,j$ in the corresponding partition (the number of boxes to the right of $(i,j)$ plus the number of boxes below $(i,j)$ plus one). The product $\Pi_{(i,j) \in \mu} h(i,j)$ factorizes as: 
\begin{align}
    \Pi_{(i,j) \in \mu} h(i,j) = \Pi_{(1,j), j \in [n-k]} h(1,j) \times \Pi_{(i,j) \in \alpha} h(i,j) =  \Pi_{(1,j), j \in [n-k]} h(1,j) \times \frac{k!}{d_\alpha}
\end{align}
where the first factor is the product of hook-lengths of boxes in the first row of $\mu$. Notice that $h(1,j) \geq n-k - j + 1$, from which: 
\begin{align}
    d_\mu \leq \frac{n! d_\alpha}{k! (n-k)!} \leq \frac{n^k d_\alpha}{k!}
\end{align}
Since $k$ is constant, then $d_\alpha$ is constant and the dimension of $\mu$ is bounded by a polynomial.
\end{proof}

\begin{corollary}
 $\mu = (n-k, \alpha_1, \alpha_2 \ldots)$ and $\nu = \lambda$ be arbitrary partitions of $n$. Then the Algorithm \ref{theorem_quantum_algo} computes the Kronecker coefficient $g_{\mu \lambda \lambda}$  in polynomial time. 
\end{corollary}
\begin{proof}
    Reorder the partitions so that the algorithms estimates $g_{\mu \lambda \lambda}$ by sampling a distribution; 
    \begin{align}
        \frac{d_\lambda}{d_\lambda d_\mu} g_{\mu \lambda \lambda} = \frac{1}{d_\mu} g_{\mu \lambda \lambda}.
    \end{align}
    We have that $d_\mu \leq n^k d_\alpha / k!$ from above, which means that we can resolve the (integer) coefficient $g_{\mu \lambda \lambda}$ with polynomially many shots (following Lemma~\ref{Lemma:shots}, see also below).
\end{proof}
 While we were updating this manuscript, Greta Panova showed that the problem of computing Kronecker coefficients on this class of partitions can be solved in polynomial time by classical algorithm \cite[Theorem 1]{panova2025polynomialtimeclassicalversus} in time $\tilde{\OC}(n^{4k^2+1})$. In the case of Kronecker coefficients, we now claim that each run of Algorithm~\ref{theorem_quantum_algo} requires $\OC(n^4)$ of elementary quantum gates. For $d_\mu \leq n^k$, the Kronecker coefficient can be estimated with high probability with at most $n^k$ many shots; counting each shot as one elementary operation then leads to $\OC(n^{2k+4})$ runtime. Assuming her algorithm is optimal, this suggest a large polynomial speedup in computing the Kronecker coefficients for the partitions in Lemma~\ref{Lemma:partitions}.

 \begin{lemma}
     The Kronecker coefficient algorithm requires $\OC(n^4)$ elementary gates per shot. 
     \label{Lemma:GateComplexityRound}
 \end{lemma}
 \begin{proof}
     The algorithm has three key components: a preparation of the two input states $\rho_{S_n}^\lambda$ and $\rho_{S_n}^\mu$ (which can be done by $S_n$-Fourier transform), a controlled application of two copies of the regular representation and an inverse $S_n$-QFT on the control register. Following the results of \cite{kawano2016quantum} $S_n$-QFT can be decomposed into $\tilde{\OC}(n^4)$-many quantum gates (if not better), which means that both preparation of the input states as well as the inverse $S_n$-QFT require $\OC(n^4)$ many gates. It remains to argue that the controlled operation of the regular representation is also $\OC(n^4)$. As discussed in \cite{kawano2016quantum}, each $\pi \in S_n$ has a unique factorization into cyclic permutations as:
     \begin{align}
         \pi &= (1,2, \ldots n)^{i_n} (1,2 \ldots n-1)^{i_{n-1}} \ldots (1,2,3)^{i_3} (1,2)^{i_2},
     \end{align}
     for integers $i_2, \ldots i_n$. This suggests that we can implement any $\pi$ by a sequence of controlled shift operators, each applied at most $n$ times. Each cyclic shift decomposes into at most $n$ transpositions. 
      Each transposition $t \in S_n$ acts as: 
      \begin{align}
          R(t) \ket{h} = \ket{th}.
      \end{align}
      Assuming that a swap gate between two qubits is an elementary operation, we now argue that a controlled action by a transposition in the regular representation is $n$. To that end, encode each permutation as a vectorized permutation matrix using $n^2$ qubits, for example: 
      \begin{align}
          (132) \mapsto \ket{010|001|100},
      \end{align}
      where the vertical $|$ separators were added to visually separate the rows of the permutation matrix. A transposition matrix acts by swapping exactly two rows (each with at most $n$ qubits), which can be implemented using up to $n$ controlled swaps. This means that the controlled action of $R(t)$ can be implemented with at most $\OC(n^3)$ (cycles $\times$ transpositions $\times$ controlled swaps) elementary gates. The leading contribution to the gate complexity is then the $S_n$-QFT, which gives ${\OC}(n^4)$ elementary gates per round.
 \end{proof}
 We thank Dmitry Grinko for pointing out the factorization from \cite{kawano2016quantum}.

 \begin{lemma}
Let $d_\lambda = d_\nu \geq d_\mu$, such that $d_\mu \leq n^k$. Then algorithm \ref{theorem_quantum_algo} uses $\OC(n^{2k + 4})$ gates and finds $g_{\mu \lambda \lambda}$ with constant probability. 
 \end{lemma}
 \begin{proof}
     Lemma~\ref{Lemma:shots} with $d_\beta := d_\lambda$, $d_\alpha = d_\lambda d_\mu$ and $\mult{r^\lambda_{S_n}}{r^\mu_{S_n} \otimes r^\lambda{S_n}} = g_{\mu \lambda \lambda}$ implies that we need $\OC(n^{2k})$-many shots to estimate the Kronecker coefficient with high probability.  Combining this with Lemma~\ref{Lemma:GateComplexityRound}, leads to $\OC(n^{2k + 4})$ gate complexity for the algorithm. 
 \end{proof}

 \paragraph{Numerical Results: The Christandl, Doran, Walter agorithm}
We studied the Christandl, Doran and Walter \cite{Christandl12} algorithm numerically and compared its performance to a naive, brute-force, implementation of Kronecker computation. The `brute-force' implementation evaluates the character formula for the Kronecker coefficient: 
\begin{align}
    g_{\mu \nu \lambda} &= \sum_{\pi \in S_n} \frac{\chi^\mu(\pi) \chi^\nu(\pi) \chi^\lambda(\pi)}{n!}  = \sum_{\gamma \in C} \frac{c(\gamma) \chi^\mu(\gamma) \chi^\nu(\gamma) \chi^\lambda(\gamma)}{n!},
\end{align}
where $C$ is the set of conjugacy classes of $S_n$, $c(\gamma)$ is the size of the conjugacy class with element $\gamma$ and $\chi^\mu(\gamma)$ is the value of the character on the conjugacy class $\gamma$. We implemented this using Sage's \textit{SymmetricGroupRepresentation} on a given conjugacy class, which we convert to character table. We assume that we precomputed sizes of the conjugacy classes and we time the evaluation of the characters and the summation over the character table. We also do not time the normalization of the character sum that gives Kronecker coefficients. As is, the algorithm is quite inefficient and one may object that the choice of not timing the conjugacy class precomputation is arbitrary. We invite the interested reader to suggest improvements to our methodology. See the stub of the Sage/Python code used for the `brute-force' implementation:  

\begin{verbatim}
def evaluate_character_formula(mu,nu, ell):
    """ Simple algorithm for the Kronecker """
    n = sum(mu)
    conjugacy_class_sizes = list(map(lambda x: len(x.list()), SymmetricGroup(n).conjugacy_classes()))

   
    start = timeit.default_timer()
    a = SymmetricGroupRepresentation(mu).to_character().values()
    b = SymmetricGroupRepresentation(nu).to_character().values()
    c = SymmetricGroupRepresentation(ell).to_character().values()
   
    # Kronecker is not normalized!
    sum([x*y*z*cc_size for x in a for y in b for z in c for cc_size in conjugacy_class_sizes])
    end = timeit.default_timer()

    return end - start
\end{verbatim}

We compared this algorithm to  the implementation at \url{https://github.com/qi-rub/barvikron/} (accessed on Jul 22 2024), running Latte Integrale \url{https://github.com/latte-int/latte-distro} (accessed on Jul 22 2024) implementation of Barvinok's algorithm. Unfortunately, we did not succeed at setting up the alternative barvinok package on the backend. We also only used a single processor implementation of the algorithm, both for the brute-force and barvikron implementation, and timed the algorithm evaluation on: 

\begin{verbatim}
  Model Name:	MacBook Pro
  Model Identifier:	Mac14,6
  Model Number:	Z17400188LL/A
  Chip:	Apple M2 Max
  Total Number of Cores:	12 (8 performance and 4 efficiency)
  Memory:	64 GB
\end{verbatim}
Fig.~\ref{fig:numerics} shows the comparison of the CDW algorithm runtime to a bruteforce implementation also timed using the timeit python package on this hardware. We only ran the CDW algorithm for inputs with up to $6$ partitions and up to $n=9$ boxes. At least for inputs with large number of parts and small number of boxes, the CDW algorithm seems to be outperformed by the naive algorithm for short (small number of boxes compared to number of parts) partitions.
\begin{figure}[h]
    \centering
    \includegraphics[scale=0.6]{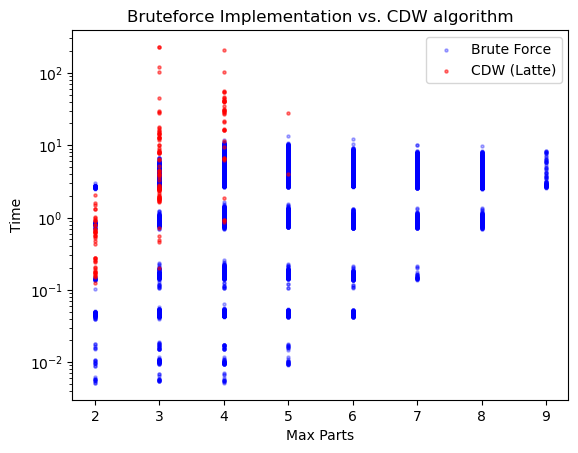}
    \caption{Christandl-Doran-Walter \cite{Christandl12} algorithm vs brute-force implementation for "short" partitions. Max number of parts vs runtime (logscale).}
    \label{fig:numerics}
\end{figure}

This is of course not the set of inputs on which (still formally polynomial-time) CDW algorithm shines: it works best for partitions of small length and (very) large number of boxes; in fact \url{https://github.com/qi-rub/barvikron/} reports evaluation for $\lm = [400000,200000,100000], \mu = [500000,100000,100000], \nu = [300000,200000,200000]$, which is clearly beyond the capability of both the bruteforce algorithm as well as the quantum algorithms proposed here (the irrep dimensions become too large). What the numerics suggest is that the CDW algorithm has a poor runtime dependence on the number of parts in the inputs. See also the discussion in Ref.~\cite{mishna2022estimating}.

\begin{figure}[h]
    \centering
    \includegraphics[scale=0.6]{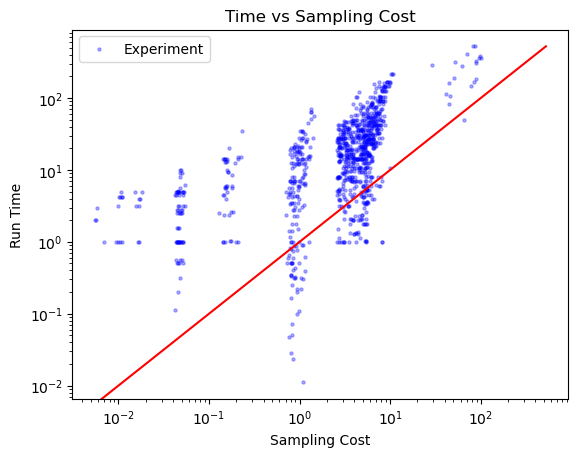}
    \caption{Sampling cost vs runtime (logscale) for subset of inputs on which Kronecker coefficient is nonzero and $3 < n < 10$ }
    \label{fig:numerics2}
\end{figure}

The brute-force algorithm will scale exponentially on the worst possible input, but this doesn't say much about how well it does on a typical input (say, sampled from a Plancherel distribution), especially in comparison to the quantum algorithm.  Assuming that we can get a relatively efficient implementation of the algorithm in Thm.~\ref{theorem_quantum_algo}, we can get some intuition by comparing the sample cost to the runtime of the brute force algorithm (see Fig.~\ref{fig:numerics2}). For a very crude argument, supported by the plot, if we count each sample as one computational step, the quantum algorithm often does quite a bit better than the classical implementation. This should be interpreted cautiously though - our implementation of the Kronecker algorithm is very far from optimal and observations like this are of course only a necessary conditions for existence of a quantum advantage. Additional work has to be done to provide more (or less!) support for Conjecture~\ref{Conjecture:Kron} and we would like to invite the interested readers to implement a more efficient algorithm for Kronecker evaluation and compare to the sampling cost of the quantum algorithm. We would like to thank S.~Bravyi, among many other things, for suggesting the plot in Fig.~\ref{fig:numerics2}.

We end this section with an interesting observation: we noticed that it becomes very challenging to sample a random (Placherel-weighted) triple of partitions for which the Kronecker coefficient became nonzero. This suggests that we can use Thm.~\ref{theorem_quantum_algo} and Thm.~\ref{theorem_quantum_algo_ind} to solve the following distributional search problem:  Given $\lambda, \mu$ sampled from the Plancherel distribution on $n$, find $\nu$ such that $g_{\lambda \mu \nu} > 0$. The hardness of finding such coefficient as we increase the input size \emph{suggests} that the problem is hard to solve by a simple classical algorithm, but a quantum algorithm solves it in polynomial time. It may be interesting to study this problem in more detail.

\subsection{Plethysm Coefficients: Details}

\paragraph{Semidirect Product Groups.} 

Given a finite group $F$ with an action on another finite group $H$, that is, provided a group homomorphism $\phi:F \arr {\rm Aut}(H)$ exists, the \textit{semidirect product} group $H \rtimes_{\phi} F$ is the group with set being the cartesian product $H\times F$, i.e., all pairs $(h,f)$ for $h\in H$ and $f\in F$, and with a group operation $\circ: (H\rtimes_{\phi} F) \times (H\rtimes_{\phi} F) \arr H\rtimes_{\phi} F$ given by
\begin{align}
&(h',e_F) \circ (h,f) = (h'\circ h,f)  \\
&(e_H,f') \circ (h,f) = (\phi(f')\circ h,f'\circ f)  \,,
\end{align}
for all $h,h'\in H$ and $f,f'\in F$, where $e_H$ and $e_F$ are used to denote the identity element on each group.

Perhaps the simplest example for a semidirect product group is given by the \textit{Dihedral group} $D_n = Z_n \rtimes_\phi Z_2$, which consists of the semidirect product between cyclic groups of order $2$ and $n$, where $Z_2=\{0,1\}$ and $Z_n=\{i\}_{i\in[n]}$, and where $\phi(0)\cdot i=i$  and $ \phi(1)\cdot i = -i \, $.

\paragraph{Wreath Product Groups.} 
A special kind of semidirect product groups, when $H=G^{\times n}$ and $F\subseteq S_n$ acts on $H$ by permuting the copies of $G$ in $H$, are called \textit{wreath product} groups, and denoted $G\wr F = G^{\times n} \rtimes S_n$. Let us focus on the case $F=S_n$. Intuitively, $G \wr S_n$ corresponds to the set of $n\times n$ permutation matrices with "ones" replaced by elements of $G$. Given $U=\spn\{\ket{u}\}_{u}$ and $V=\spn\{ \ket{v} \}_v$ $G$ and $S_n$ modules respectively, the wreath product $U\wr V$ is a $G\wr S_n$ module with vector space $U\tn \otimes V$ and group action

\begin{align}
    &(g_1,\cdots,g_n,e_{S_n})\cdot \ket{u_1}\otimes \cdots \otimes \ket{u_n} \otimes \ket{v} = \Big(g_1\cdot \ket{u_1}\Big) \otimes \cdots \otimes \Big(g_n \cdot\ket{u_n}\Big) \otimes \ket{v}\, \text{ and }\\
    &(e_G,\cdots,e_G,\sg) \cdot \ket{u_1}\otimes \cdots \otimes \ket{u_n} \otimes \ket{v} = \ket{u_{\sg^{-1}(1)}} \otimes \cdots \otimes \ket{u_{\sg^{-1}(n)}} \otimes \Big( \sg\cdot\ket{v} \Big)\,.
\end{align}

In fact, one can show $U\wr V$ is irreducible whenever $U$ and $V$ are~\cite{stembridge1989eigenvalues}, yet not all irreps of $G\wr S_n$ come from the wreath of irreps of $G$ and $S_n$. Although we will only care for such subset of irreducibles, let us mention that the strategy to obtain a complete set of irreps for $G\wr S_n$ consists of considering representations of $G\wr S_n$ that are induced from certain representation of its Young-subgroup-like subgroups $G\wr S_\lm$ for $\lm\vdash n$. We will not need this here so we refer the reader to Ref.~\cite{stembridge1989eigenvalues} for further details.

\paragraph{Plethysm Coefficients.}
\label{App:Plethysms} 

Consider the irreducible decomposition of the reducible $\SU(d)$ module obtained by composing Schur functors $\mbb{S}^\lm$ and $\mbb{S}^\mu$ on $\mbb{C}^d$
\begin{align}
    \mbb{S}^\lm ( \mbb{S}^\mu (\mbb{C}^d)) \cong \bigoplus_\nu \mbb{C}^{a_{\lm\mu}^\nu}\otimes \mbb{S}^\nu(\mbb{C}^d)
\end{align}
where $a_{\lm\mu}^\nu$ are the Plethysm coefficients. The computational complexity of the map $a: \PC_{n,d}^{\times 3} \arr \mbb{N}$ was studied by Fischer and Ikenmeyer~\cite{Fischer2020}. 
Most importantly (at least for our purposes), the plethysm coefficients also appear as multiplicities in the restriction of $S_n$ irreps to the wreath product subgroup $S_a\wr S_b \subset S_{ab}$ (where $ab=n$), i.e., the subgroup with set $S_a^{\times b} \times S_b$ with $S_b$ permuting the copies of $S_a$:

\begin{equation}
a^\nu_{\lm\mu} = \mult{r_{S_a\wr S_b}^{\lm\mu}}{ r_{S_{ab}}^\nu \downarrow^{S_{ab}}_{S_a\wr S_b}}\,,
\end{equation}
where $r_{S_a\wr S_b}^{\lm\mu}$ denotes the $S_a\wr S_b$-irrep  $r_{S_a}^\lm \wr r_{S_b}^\mu$ of dimension $d_\lm^b d_\mu$.

\subsection{Induction Algorithm}

Our induction algorithm uses a subgroup register $\mc{S}=\mbb{C}[H]=\spn\{\ket{h}\}_{h\in H}$ and a coset register $\mc{V}=\mbb{C}[G/H]=\spn{\ket{t_i}}_{i\in[N]}$.
We begin with a maximally-mixed state over some isotypic in $\mc{S}$ 
\begin{align}
    \rho^\b_{H} &= 
V^\dag \left( \ket{\b}\bra{\b} \otimes \frac{I_{d_\b}}{d_\b} \otimes \frac{I_{d_\b}}{d_\b}\right) V = \Pi_\b^H / d_\b^2,
\end{align}
in the subgroup register $\mc{S}$, where $V$ is the $H$-QFT. This could be prepared with $H$-QFT or given access to some oracle.

We make use of an embedding unitary $U_E$ acting on $\mc{V}\otimes \mc{S}$ mapping a pair $(t,h)$ of transversal and subgroup elements to the corresponding $g = th$
\begin{align}
    U_E \Big( \ket{t}_{\mc{V}} \otimes \ket{h}_{\mc{S}} \Big) = \ket{g}.
\end{align}
The mapping $(t,h)\arr g$ is unique and the above transformation is unitary. It can be efficiently implemented as a quantum circuit as long as there is an efficient classical circuit for group multiplication. This way, we can ``match the bases'' of the $H$-regular representation and the $G$-regular representation. Note that Ref.~\cite{krovi2019efficient} used similar trick for QFT over induced representations. Let: 

\begin{align}
     R(g) &= \sum_{a \in G} \ket{ga}\bra{a},  & S(h) &= \sum_{b \in H} \ket{hb}\bra{b}, & I_N &= \sum_{c \in |G|/|H|} \ket{c}\bra{c}. 
 \end{align}
 
As depicted in Fig.~\ref{fig_2:induction}, we initialize the state $\rho^\b_H \otimes I_N/N$ in $\mc{S}\otimes\mc{V}$, then apply the embedding unitary and finally weak Fourier sample $G$-irreps. The probability of obtaining $G$-irrep $\a$ is given by 
\begin{align}
p_\beta(\alpha) &= \frac{1}{N}\Tr \left[ \Pi^\a_G U_E (\rho^\beta_H  \otimes I_N ) U_E^\dag \right]
= \frac{d_\a d_\b}{N |G||H| d_\b^2} \sum_{g\in G} \sum_{h\in H} \chi_G^\a(g^{-1}) \chi_H^\b(h) \Tr[R(g) U_E (S(h^{-1})\otimes I_N)U_E\ad]\,,
\end{align}
where we've used that for a finite group $F$ and a representation $T$, the projector into its $\g$-isotypic subspace is \cite{serre1977linear}: 
\begin{align} \Pi^\g_F=\frac{d_\g}{|F|} \sum_{f\in F} \chi^\g_F(f^{-1}) T(f).\end{align}
We now unwrap $\Tr[R(g)   U_E (S(h^{-1}) \otimes I_N) U_E^\dag]$ to get  $U_E (S(h^{-1}) \otimes 1) U_E^\dag = \sum_{b \in H} \sum_{i \in [N]} \ket{h^{-1}bt_i} \bra{bt_i}\,$. Notice that $bc$ is an element of $G$ and that the register has dimension $|G|$ and we assume that we match the basis of $R(g)$.  It follows from orthogonality that: 
\begin{align}
     \Tr[R(g)   U_E (S(h^{-1}) \otimes I_N) U_E^\dag]  &=  \sum_{a \in G} \sum_{b \in H} \sum_{i\in [N]} \braket{bt_i}{ga} \braket{a}{h^{-1}bt_i} \\ &= \sum_{b \in H} \sum_{i \in[N]} \braket{g^{-1}bt_i}{h^{-1}b t_i} = |G| \delta_{hg}.
\end{align}
Plugging this back gives: 
\begin{align}
p_\beta(\alpha) &= \frac{d_\a}{N d_\b |H|} \sum_{h\in H} \chi_G^\a(g^{-1}) \chi_H^\b(h) \\
&=\frac{d_\a }{Nd_\b} \mult{r^\b_H}{r^\a_G \downarrow^G_H} \\
&=\frac{d_\a }{Nd_\b} \mult{r^\a_G}{r^\b_H \uparrow^G_H}\,. \\
\end{align}
where the last line follows from Frobenius's reciprocity.

\subsection{Connection to the Hidden Subgroup Problem}
We briefly remark that Thm.~\ref{theorem_quantum_algo} has an interesting relationship to the standard approach to the Hidden Subgroup Problem (HSP). In HSP, one aims to find a subgroup $H \subseteq G$ using oracle access to a function $f_H:G \arr [N]$ for $N = |G|/|H|$ that takes distinct values on cosets of $H$ in $G$:
\begin{equation}
    G = \bigcup_{i=1}^N t_i H\quad \text{ and }\quad f(t_iH)=i\,,
\end{equation}
where $\{t_i\}$ is a set of $N$ elements each chosen from a distinct coset $i$ of $H$ in $G$. 
The {standard approach} to HSP prepares the state $\sum_{g \in G} \ket{g}_{\mathcal{C}} \otimes \ket{f(g)}_{\mathcal{T}}$ and ignores $\mathcal{T}$ to obtain the {hidden subgroup state}:
\begin{equation}
    \rho_{G}^{G/H} =\frac{|H|}{|G|} \sum_{i=1}^N \ketbra{t_i H}{t_i H},
\end{equation}
where $\ket{t_i H} =  \sum_h \ket{t_i h} / \sqrt{|H|}$ is the \emph{coset state}.  $\rho^{G/H}_G$ is the maximally-mixed state over the subspace $\mbb{C}[G/H] =\spn\{ \ket{t_i}\}_{i=1}^N \subseteq \mbb{C}[G]$ of left cosets of $H$ on $G$. Such subspace $\mbb{C}[G/H]$ is a $G$-representation that is induced from the trivial irrep of $H$, $r^\triv_H\uparrow_H^G$. One approach to take from here is to weak Fourier-sample: apply $G$-QFT to $\rho_{G}^{G/H}$ and measure the $G$-irrep register, obtaining the output irrep $\a$ of $G$ with probability: 
\begin{equation}
\begin{aligned}
    p_{G/H}(\a) &= \mult{r^\a_G }{r^\triv_H\uparrow_H^G} \frac{d_\a}{\text{dim}(r^\triv_H\uparrow_H^G) } = \mult{r^\a_G }{r^\triv_H\uparrow_H^G} \frac{d_\a |H|}{|G| }.
    \label{Eq:HSP}
    \end{aligned}
\end{equation}
For $G=S_n$ and $H = S_\mu$, 
one obtains $S_n$-irrep $\nu$ with probability $p_{S_n/S_\mu}(\nu) =K^\mu_\nu d_\nu / (n!/\mu!)$.
Compare this with the distribution sampled in Thm.~\ref{theorem_quantum_algo}, which would be $K_\nu^\mu / d_\nu$.

\subsection{Weak Fourier Sampling of the Hidden Subgroup State}

For reference, we show how to arrive at Eq.~\eqref{Eq:HSP}. We have that: 
\begin{align}
   p_{G/H}(\alpha) &= \Tr\left(\rho_{G}^{G/H} \Pi_\alpha \right) = \frac{|H|}{|G|} \sum_{i=1}^N \Tr\left( \ket{t_i H}\bra{t_iH} \Pi_\alpha \right) = \frac{|H|}{|G|}\sum_{i=1}^N \alpha_i,
\end{align}
so by linearity, we can just analyze $\alpha_i$: 
\begin{align}
    \alpha_i &=  \frac{1}{|H|} \sum_{k, \ell \in [{d_\alpha}]} \sum_{h,h' \in H}  \braket{t_i h}{U | \alpha, k, \ell} \braket{\alpha, k, \ell}{U^\dag | t_i h'},
\end{align}
where $U$ is the $G$-QFT with matrix elements (see also Eq.~\ref{Eq:QFT_matrix_elements}): 
\begin{align}
  U\ket{g} &= \sum_{\alpha} \sqrt{\frac{d_\alpha}{|G|}} \sum_{ij} \braket{i}{r^\alpha_G(g)|j}\ket{\alpha, i, j},
\end{align}
This gives:
\begin{align}
    \alpha_i &= \frac{d_\alpha}{|H| |G|} \sum_{k, \ell} \sum_{h,h' \in H}  \Tr(r_G^\alpha(t_i h) r_G^\alpha(t_i h')^\dag) \\ &= \frac{d_\alpha}{|H| |G|} \sum_{k, \ell} \sum_{h,h' \in H}  \Tr(r_G^\alpha(t_i h h'^{-1} t_i^{-1}) ) \\ 
    &= \frac{d_\alpha}{|G|}  \sum_{h \in H} \chi_G^\alpha(h),
\end{align}
where $\chi_G^\alpha$ the irreducible character $\alpha$ of $G$. Character orthogonality leads to: 
\begin{align}
\frac{1}{|H|}\sum_{h \in H} \chi_G^\alpha(h) = (\chi_G^\alpha, \chi^{\triv}_H)_H = \mult{r^{\triv}_H}{r^\alpha \downarrow^G_H}. \end{align} It follows that: 
\begin{align}
\Tr\left(\rho_{H} \Pi_\alpha \right) &= \frac{|H|d_\alpha}{|G|}\mult{r^{\triv}_H}{r^\alpha_G \downarrow^G_H}.
\end{align}

\end{document}